\documentclass[journal]{IEEEtran}

\usepackage{amsmath} 
\usepackage{tikz}
\usepackage{fix-cm}
\usetikzlibrary{arrows}

\usepackage{graphicx}

\usetikzlibrary{calc,positioning}
\usepackage[outline]{contour}
\contourlength{1.5pt}
\usepackage{amsthm}

\newcommand{\vect}[1]{\mathbf{#1}}

\newcommand{\matr}[1]{\mathbf{#1}}

\newcommand{\junk}[1] {}

\def\XXint#1#2#3{{\setbox0=\hbox{$#1{#2#3}{\int}$}
\vcenter{\hbox{$#2#3$}}\kern-.5\wd0}}

\newcommand*\widebar[1]{%
  \hbox{%
    \vbox{%
      \hrule height 0.5pt 
      \kern0.3ex
      \hbox{%
        \kern-0.05em
        \ensuremath{#1}%
        \kern-0.05em
      }%
    }%
  }%
}

\usepackage [english]{babel}
\usepackage [autostyle, english = american]{csquotes}
\MakeOuterQuote{"}
\usepackage{cuted}
\usepackage{soul}
\usepackage{xcolor}

\newcommand{\added}[1]{{\color{black}#1}}

\newcommand{\half}{\frac{1}{2}}
\newcommand{\ihat}{\hat{\imath}}
\newcommand{\jhat}{\hat{\jmath}}

\newcommand{\E}{\mathcal{E}}

\newcommand{\dx}{\Delta x}
\newcommand{\dxhat}{\Delta \hat{ x}}
\newcommand{\dyhat}{\Delta \hat{ y}}
\newcommand{\dy}{\Delta y}

\newcommand{\dt}{\Delta t}

\newcommand{\Dl}{\matr{D}_{l}}
\newcommand{\Dlp}{\matr{D}_{l'}}

\newcommand{\eps}{\varepsilon}

\newcommand{\Deb}{\matr{D}_{\varepsilon}}

\newcommand{\Dseb}{ \matr{D}_{\sigma}}

\newcommand{\Shat}{\hat{\matr{S}}}
\newcommand{\T}{\vect{T}}

\newcommand{\Hzhat}{\vect{\hat{H}}_z}

\newcommand{\ESleft}{\frac{\dy}{2} \left( \frac{\eps_x}{\dt} + \frac{\sigma_x}{2} \right)}
\newcommand{\ESright}{\frac{\dy}{2} \left( \frac{\eps_x}{\dt} - \frac{\sigma_x}{2} \right)}

\newcommand{\Ehat}{\hat{\vect{E}}}
\newcommand{\Hhat}{\hat{\vect{H}}}

\newcommand{\B}{\hat{\vect{B}}}
\newcommand{\R}{\hat{\vect{R}}}
\newcommand{\F}{\hat{\vect{F}}}
\newcommand{\Lmatr}{\hat{\vect{L}}}
\newcommand{\xhat}{\hat{\vect{x}}}
\newcommand{\yhat}{\hat{\vect{y}}}
\newcommand{\uhat}{\hat{\vect{u}}}

\newcommand{\Rt}{\tilde{\vect{R}}}
\newcommand{\Ft}{\tilde{\vect{F}}}
\newcommand{\Bt}{\tilde{\vect{B}}}
\newcommand{\Lmatrt}{\tilde{\vect{L}}}
\newcommand{\xt}{\tilde{\vect{x}}}

\newcommand{\Nxhat}{\hat{N}_x}
\newcommand{\Nyhat}{\hat{N}_y}

\newtheorem{theorem}{Theorem}

\definecolor{HV}{rgb}{0,0.3,0}
\newcommand{\hvcolorname}{green}
\definecolor{Hcol}{rgb}{0.7,0,0}
\definecolor{EB}{rgb}{0.7,0,0}
\definecolor{Ecol}{rgb}{0,0,1}

\definecolor{legendGreen}{rgb}{0,0.5,0}
\definecolor{legendRed}{rgb}{1,0,0}
\definecolor{legendBlue}{rgb}{0,0,1}
\definecolor{legendGrey}{rgb}{0.5,0.5,0.5}
\definecolor{legendPink}{rgb}{1 0.2 1}
\definecolor{legendPurple}{rgb}{0.5 0 1}
\definecolor{legendLightgreen}{rgb}{0.47 0.67 0.19}

\begin{document}
%
\title{A Stable FDTD Method with Embedded Reduced-Order Models}

%

\author{ Xinyue~Zhang,~\IEEEmembership{Student Member,~IEEE},
	Fadime~Bekmambetova,~\IEEEmembership{Student Member,~IEEE},  and~Piero~Triverio,~\IEEEmembership{Senior Member,~IEEE}
\thanks{Manuscript received ...; revised ...}%
\thanks{This work was supported in part by the Natural Sciences and Engineering Research Council of
Canada (Discovery grant program) and in part by the Canada Research Chairs program.}
\thanks{X.~Zhang, F.~Bekmambetova, and P.~Triverio are with the Edward S. Rogers Sr. Department of Electrical and Computer Engineering, University of Toronto, Toronto, M5S 3G4 Canada (email:  xinyuezhang.zhang@mail.utoronto.ca, fadime.bekmambetova@mail.utoronto.ca, piero.triverio@utoronto.ca).}
\thanks{\textcopyright 2017 IEEE.  Personal use of this material is permitted.  Permission from IEEE must be obtained for all other uses, in any current or future media, including reprinting/republishing this material for advertising or promotional purposes, creating new collective works, for resale or redistribution to servers or lists, or reuse of any copyrighted component of this work in other works.}
}

\markboth{IEEE Transactions on Antennas and Propagation}%
{Zhang, Bekmambetova, Triverio}
%

\maketitle

\IEEEpeerreviewmaketitle
\begin{abstract} 
The computational efficiency of the Finite-Difference Time-Domain (FDTD) method can be significantly reduced by the presence of complex objects with fine features. Small geometrical details impose a fine mesh and a reduced time step, significantly increasing computational cost. Model order reduction has been proposed as a systematic way to generate compact models for complex objects, that one can then instantiate into a main FDTD mesh. However, the stability of FDTD with embedded reduced models remains an open problem. We propose a systematic method to generate reduced models for FDTD domains, and embed them into a main FDTD mesh with \added{\emph{guaranteed stability} up to the Courant-Friedrichs-Lewy (CFL) limit of the fine mesh. With a simple perturbation technique, the CFL of the whole scheme can be further extended beyond the fine grid's CFL limit}. Reduced models can be created for arbitrary domains containing inhomogeneous and lossy materials. Numerical tests confirm the stability of the proposed method, and its potential to accelerate multiscale FDTD simulations. 
\end{abstract}
\begin{IEEEkeywords}
	Finite-Difference Time-Domain method, model order reduction, subgridding, CFL limit extension, stability, passivity.
\end{IEEEkeywords}

\def\figurename{Fig.}

\section{Introduction}
\label{sec:intro}
The Finite-Difference Time-Domain (FDTD) method\added{, also known as the Yee scheme,} solves Maxwell's equations numerically using a staggered grid, where electric and magnetic fields are sampled with alternate nodes~\cite{yee,Gedney}. This configuration is particularly convenient to discretize the curl operators in Maxwell's equations. An attractive feature of \added{the Yee scheme} is the use of an explicit leap-frog scheme to march in time. This approach significantly reduces the cost per iteration over implicit alternatives, but makes stability conditional to the well-known Courant-Friedrichs-Lewy (CFL) condition~\cite{Gedney}.
In 2D, the CFL condition reads
\begin{equation} \label{eq:cfl_limit}
\dt < \frac{1}{c\sqrt{\frac{1}{\dx^2}+\frac{1}{\dy^2}}} \,,
\end{equation}
where $\dt$ is time step, $\dx$ and $\dy$ are the cell dimensions along the $x$~and $y$ axes, and $c$ is the wave velocity in the medium~\cite{Gedney}. 

Unfortunately, in spite of its low cost per iteration, \added{the Yee scheme} can become quite time-consuming when  applied to large multiscale problems with detailed objects. Two issues contribute to the increase of \added{the computational cost of the Yee scheme}:
\begin{enumerate}
\item[$a)$] small details impose a refined mesh. This increases memory consumption and cost per iteration;
\item[$b)$] a refined mesh implies a shorter time step because of CFL limit~\eqref{eq:cfl_limit}. 
\end{enumerate}

Numerous solutions have been devised to increase the efficiency \added{of the Yee scheme} in handling complex objects. Subcell models have been proposed for wires~\cite{boonzaaler1992radiation}, wire bundles~\cite{berenger2000multiwire}, sheets~\cite{maloney1992efficient}, slots~\cite{ma1997comparison}, and other common structures. These models are typically efficient, but restricted to specific objects. In FDTD subgridding, a fine mesh is used to resolve small objects, while a coarse grid is used elsewhere~\cite{thoma1996consistent,xiao2007three}. Subgridding reduces the growth of the number of unknowns caused by mesh refinement (issue $a)$ in the list above). However, it does not help with the second issue, since the fine grid imposes its CFL limit on the whole domain. While different time steps can be used in the coarse and fine grids, ensuring the stability and accuracy of the resulting scheme is quite challenging~\cite{wang2010analysis}. 

\added{One of the most significant approaches to bypass the CFL limit is represented by implicit methods, such as the Alternating-Direction Implicit FDTD (ADI-FDTD) method~\cite{adi-fdtd-2, adi-fdtd,lee2004some}. Being unconditionally stable for any time step, implicit methods are particularly attractive when combined with FDTD subgridding. In the schemes proposed in~\cite{WangADIFDTDsubgridding,ChenADIFDTDsub3D,YangADIFDTDsub2D}, the fine grid is updated with ADI-FDTD, while the coarse grid is updated with the Yee scheme, taking advantage of the respective strengths of implicit and explicit methods. This hybrid ADI-FDTD subgridding scheme allows a larger time step throughout the entire domain, but leads to an increase of the cost per iteration, as our results will show.}

Model order reduction (MOR) has been proposed to generate compact FDTD models for complex objects~\cite{cangellaris1999rapid,denecker,kulas2003reduced}. The underlying idea is to initially use a fine mesh to resolve complex objects, and then compress the FDTD equations using MOR. The reduced model is finally embedded into the surrounding coarse grid~\cite{denecker}. Typically, the insertion of the reduced model in the coarse grid reduces the CFL limit. However, because of its low order, the reduced model can be manipulated to extend the CFL limit, and enable the use of a larger time step throughout the entire computational domain~\cite{cnf-2015-nemo-fdtdmor,SpatialFiltering,DJiao}. Therefore, an MOR-based approach can be used to tackle both issues $a)$ and $b)$ caused by mesh refinement.

While MOR clearly has the potential to accelerate multiscale FDTD simulations, its use has been so far limited by stability considerations. When a reduced model is embedded into a main FDTD grid, it can easily lead to instability. Unfortunately, this happens even if the reduced model, by itself, is stable~\cite{denecker}. Kulas and Mrozowski~\cite{kulas2004stability} investigated the stability of an FDTD scheme with \added{reduced-order models}, and derived a reciprocity condition for stability. An expression for the CFL limit of the resulting scheme is also provided. There are two main limitations in this work. First, the derivation is valid only for the lossless case. Second, the expression for the CFL limit requires the norm of a very large matrix. For realistic problems, such norm cannot be computed. An upper bound can be derived, but leads to a conservative CFL limit. Moreover, the estimation becomes increasingly involved when multiple reduced models are embedded into the main grid.

In this paper, we propose a systematic theory to generate reduced models for FDTD regions, and couple them to a main FDTD grid \added{in a stable way. 
In the proposed method, the regions with complex objects are initially meshed with a refined grid. The FDTD equations for those regions are written in the form of a discrete-time dynamical system with suitable inputs and outputs, which is then compressed with MOR. The obtained reduced model is finally embedded into the surrounding coarse grid. The stability of the final scheme is guaranteed by construction up to the CFL limit of the fine grid, as we rigorously prove. We also show that the stability limit of the proposed scheme can be extended beyond the fine grid's CFL limit with a simple perturbation of the reduced model coefficients. Some preliminary results about the proposed method were presented in~\cite{cnf-2016-epeps-fdtdmor}.}

\added{There are several novel aspects in this work:
\begin{enumerate}
	\item we propose the first MOR method for the Yee scheme whose reduced models are guaranteed to be stable even when coupled to another FDTD grid. This feature enables the application of MOR to subdomains of a large problem, for much higher scalability. Previous MOR methods with guaranteed stability~\cite{cnf-2014-ims-fdtd,cnf-2014-aps-fdtd,jnl-2014-tmtt-fdtd} can only be applied to the entire domain, which is not feasible for large problems;
	\item stability is guaranteed by construction and rigorously proved. The proof is valid for an arbitrary number of reduced models embedded in a surrounding grid. This feature is a significant advantage over previous works that do not provide a theoretical stability analysis~\cite{WangADIFDTDsubgridding,SpatialFiltering,cnf-2015-nemo-fdtdmor}, or neglect losses and give stability conditions that become increasingly complicated as the number of reduced models grows~\cite{kulas2004stability}; 
	\item the CFL limit of the resulting scheme is known a priori, and is at least equal to the CFL limit of the fine grid;
	\item the CFL limit of the final scheme can be extended beyond the fine grid's limit, while requiring only explicit update operations at runtime.
\end{enumerate}
}
The numerical examples in Sec.~\ref{sec:numerical_examples} confirm the stability of the proposed method, \added{and demonstrate its computational efficiency for multiscale problems in comparison to the Yee scheme, FDTD subgridding, and an ADI-FDTD subgridding scheme.}

\section{FDTD Equations for the Fine Grid as a Discrete-Time Dynamical System}
\label{sec:fdtdsystem}

We consider a multiscale scenario where several objects with fine features are present in a large two-dimensional domain. As in subgridding, we use a fine mesh in the regions with fine features. A coarse mesh is used in the rest of the domain. Without loss of generality, we consider the case of a single region with fine mesh. The extension to multiple regions is straightforward since the same process can be independently applied to each refined region. The cell dimensions in the coarse grid are denoted as~$\dx$ and~$\dy$. In the fine region, the cell dimensions are
\begin{gather}
	\dxhat =  \frac{\dx}{r}\,, \qquad
	\dyhat = \frac{\dy}{r}\,,
\end{gather}
where the refinement factor $r>1$ is an integer. Throughout the paper, we use symbols with a hat for the quantities associated with the fine grid, and symbols with no hat for the quantities associated with the coarse grid. 
In order to later apply MOR, we cast the FDTD equations for the fine region in the form of a discrete-time dynamical system, having the E and H field tangential to the region boundaries as outputs and inputs~\cite{bekmambetova2016dissipative}. These variables will allow us to reconnect the reduced model of the fine grid to the main grid.

\subsection{FDTD Update Equations for the Nodes of the Fine Mesh}

We consider a rectangular fine mesh with $\Nxhat$ cells along $x$ and $\Nyhat$ cells along $y$. A graphical illustration of a simple $2 \times 2$ region is provided in Fig.~\ref{fig:cell}. We consider a TE$_z$ mode with field components $E_x$, $E_y$ and $H_z$. In addition to the standard field samples used in FDTD, we also take as variables the magnetic field on the four boundaries of the region, as shown in Fig.~\ref{fig:cell}. These samples are called hanging variables~\cite{venkatarayalu2007stable}, and serve two purposes. First, they allow us to write a self-contained mathematical model for the region, which does not involve any field sample beyond its boundaries~\cite{bekmambetova2016dissipative}. Second, hanging variables will facilitate the connection of the fine region model to the surrounding coarse grid.

\begin{figure}[t]
	\centering
	\begin{tikzpicture}
	[sgrid/.style={dotted},
	arrE/.style={very thick,->,Ecol,>=stealth,shorten >=7pt, shorten <=7pt},
	labelH/.style={font=\fontsize{\sizefont}{\sizefont}\selectfont, Hcol},
	labelHV/.style={font=\fontsize{\sizefont}{\sizefont}\selectfont, HV},
	arr/.style={thick,->,>=stealth,shorten >=7pt, shorten <=7pt},
	labelE/.style={font=\fontsize{\sizefont}{\sizefont}\selectfont, Ecol},
	label/.style={font=\fontsize{\sizefont}{\sizefont}\selectfont},scale = 1.2
	]
	\pgfmathsetmacro{\sizefont} {9};
	\pgfmathsetmacro{\W}{1.4};
	\pgfmathsetmacro{\H}{1.2};
	\pgfmathsetmacro{\extraW}{0.8};
	\pgfmathsetmacro{\extraH}{0.7};
	\pgfmathsetmacro{\Rbig}{0.08};
	\pgfmathsetmacro{\Rsmall}{0.03};
	
	
	\draw[] (0,0) rectangle (2*\W, 2*\H);
	\draw (\W,0) -- (\W,2*\H);
	\draw (2*\W,0) -- (2*\W,2*\H);
	\draw (0,\H) -- (2*\W,\H);
	\draw (0,2*\H) -- (2*\W,2*\H);

	\draw [sgrid] (0, 0.5*\H) -- (1.8*\W,0.5*\H);
	\draw [sgrid] (0, 1.5*\H) -- (1.8*\W,1.5*\H);
	\draw [sgrid] (0.5*\W, 0) -- (0.5*\W,1.8*\H);
	\draw [sgrid] (1.5*\W, 0) -- (1.5*\W,1.8*\H);
	
	\draw [arrE] (0,0) -- (1*\W, 0);
	\draw [arrE] (1*\W,0) -- (2*\W, 0);
	\draw [arrE] (0,\H) -- (1*\W, \H);
	\draw [arrE] (1*\W,\H) -- (2*\W, \H);
	\draw [arrE] (0,0) -- (0, \H);
	\draw [arrE] (\W,0) -- (\W, \H);
	\draw [arrE] (2*\W,0) -- (2*\W, \H);
	\draw [arrE] (0,\H) -- (0, 2*\H);
	\draw [arrE] (\W,\H) -- (\W, 2*\H);
	\draw [arrE] (2*\W,\H) -- (2*\W, 2*\H);
	
	\draw [Hcol, fill=white](\W/2,\H/2) circle (\Rbig);
	\draw [Hcol, fill=white](1.5*\W,\H/2) circle (\Rbig);
	\draw [Hcol, fill=white](\W/2,1.5*\H) circle (\Rbig);
	\draw [Hcol, fill=white](1.5*\W,1.5*\H) circle (\Rbig);
	
	\draw [fill,Hcol] (\W/2,\H/2) circle (\Rsmall);
	\draw [fill,Hcol] (1.5*\W,\H/2) circle (\Rsmall);
	\draw [fill,Hcol] (\W/2,1.5*\H) circle (\Rsmall);
	\draw [fill,Hcol] (1.5*\W,1.5*\H) circle (\Rsmall);
	
	\draw [HV, fill=white](0.5*\W,0*\H) circle (\Rbig);
	\draw [HV, fill=white](1.5*\W,0*\H) circle (\Rbig);
	\draw [HV, fill=white](0*\W,0.5*\H) circle (\Rbig);
	\draw [HV, fill=white](0*\W,1.5*\H) circle (\Rbig);
	\draw [HV, fill=white](0.5*\W,2*\H) circle (\Rbig);
	\draw [HV, fill=white](1.5*\W,2*\H) circle (\Rbig);
	\draw [HV, fill=white](2*\W,0.5*\H) circle (\Rbig);
	\draw [HV, fill=white](2*\W,1.5*\H) circle (\Rbig);
	\draw [HV, fill](0.5*\W,0*\H) circle (\Rsmall);
	\draw [HV, fill](1.5*\W,0*\H) circle (\Rsmall);
	\draw [HV, fill](0*\W,0.5*\H) circle (\Rsmall);
	\draw [HV, fill](0*\W,1.5*\H) circle (\Rsmall);
	\draw [HV, fill](0.5*\W,2*\H) circle (\Rsmall);
	\draw [HV, fill](1.5*\W,2*\H) circle (\Rsmall);
	\draw [HV, fill](2*\W,0.5*\H) circle (\Rsmall);
	\draw [HV, fill](2*\W,1.5*\H) circle (\Rsmall);

	\draw [labelH] (0.5*\W, 0.5*\H) node [below] {\contour{white}{$H_z|_{\frac{3}{2},\frac{3}{2}}$}};
	\draw [labelHV](0.35*\W,0*\H) node [below] {$H_z|_{\frac{3}{2},1}$};
	\draw [labelHV](0*\W,0.5*\H) node [left] {$H_z|_{1,\frac{3}{2}}$};
	

	\draw [labelE](0.95*\W,0*\H) node [below] {$E_x|_{\frac{3}{2},1}$};
	\draw [labelE](0.5*\W,1*\H) node [below] {$E_x|_{\frac{3}{2},2}$};
	\draw [->,>=stealth, Ecol] (-0.3*\W,0.8*\H) -- (-0.05*\W,0.7*\H);
	\draw [labelE](-0.3*\W,0.9*\H) node [left] {$E_y|_{1,\frac{3}{2}}$};

	\draw (0,0) node [left] {$(1,1)$};
	\draw (2*\W, 2*\H) node [above right] {\contour{white}{$(\Nxhat+1, \Nyhat+1)$}};
	
	\draw [<->, >=stealth, shorten >=2pt, shorten <=2pt] (0,2*\H) -- (\W,2*\H);
	\draw [<->, >=stealth, shorten >=2pt, shorten <=2pt] (2.2*\W,0) -- (2.2*\W,\H);
	
	\draw (\W/2,2*\H) node [above] {$\dxhat$};
	\draw (2.2*\W,0.5*\H) node [right] {$\dyhat$};
	
	\draw (2.5*\W+\extraW, 0) node [below left] {South boundary};
	\draw (0, 2*\H+0.5*\extraH) node [above right] {North boundary};
	\draw (0, 2*\H+\extraH) node [above left, rotate=90] {\begin{tabular}{c} West\\boundary\end{tabular}};
	\draw (2*\W+\extraW, \H+0*\extraH) node [above, rotate=270] {East boundary};

	\pgfmathsetmacro{\axisOffsetx}{0.7*\W};
	\pgfmathsetmacro{\axisOffsety}{\H};
	\draw[->,black,>=stealth] (2.5*\W+\extraW+\axisOffsetx, 0+\axisOffsety) -- (2.5*\W+\extraW+\axisOffsetx, 0.3*\H+\axisOffsety);
	\draw[->,black,>=stealth] (2.5*\W+\extraW+\axisOffsetx, 0+\axisOffsety) -- (2.5*\W+\extraW+\axisOffsetx + 0.3*\H, 0+\axisOffsety);
	\draw [fill=white] (2.5*\W+\extraW+\axisOffsetx, \axisOffsety) circle (0.05*\W);
	\draw [fill=black] (2.5*\W+\extraW+\axisOffsetx, \axisOffsety) circle (0.02*\W);
	\draw[label] (2.5*\W+\extraW+\axisOffsetx + 0.3*\H, 0+\axisOffsety) node [right] {x};
	\draw[label] (2.5*\W+\extraW+\axisOffsetx, 0+\axisOffsety+ 0.3*\H) node [above] {y};
	\draw[label] (2.5*\W+\extraW+\axisOffsetx, 0+\axisOffsety) node [below left] {z};
	
	\end{tikzpicture}
	\caption{Graphical illustration of a simple $2 \times 2$ fine region. The solid lines denote the primary grid. The \hvcolorname~color denotes the hanging variables introduced on the four boundaries.}
	\label{fig:cell}
\end{figure}
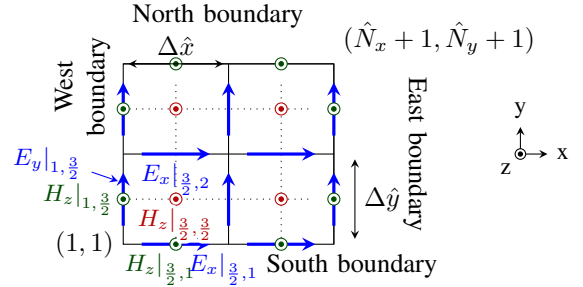
\subsubsection{Update Equations for the $E_x$ and $E_y$ Nodes}

For the~$E_x$~nodes that are strictly inside the fine region, we write a standard FDTD update equation~\cite{Gedney}
\begin{multline}
\dxhat \dyhat \left( \frac{\eps_x}{\dt} + \frac{\sigma_x}{2} \right) E_x|_{i+\half,j}^{n+1} = 
\dxhat \dyhat \left( \frac{\eps_x}{\dt} - \frac{\sigma_x}{2} \right) E_x|_{i+\half,j}^{n} \\
+ \dxhat H_z|_{i+\half,j+\half}^{n+\half} - \dxhat H_z|_{i+\half,j-\half}^{n+\half} \,,
\label{eq:updateEx}
\end{multline}
where $\eps_x$ and $\sigma_x$ are, respectively, the average permittivity and conductivity on the~$x$-oriented edges of the primary grid\footnote{To keep expressions reasonably compact, we do not indicate explicitly the dependence of $\eps_x$, $\sigma_x$, $\eps_y$, $\sigma_y$ and $\mu$ on $i$ and $j$.}. 
For the nodes on the South and North boundaries of the region, a modified FDTD equation must be used~\cite{bekmambetova2016dissipative}. A conventional FDTD equation would otherwise involve magnetic fields beyond the region boundaries, which may not even be available because of the different resolution used in the surrounding coarse grid. This issue can be avoided by using the hanging variables to approximate the spatial derivatives of the magnetic field on the boundaries. For instance, the following update equation is written for the $E_x$ nodes on the South boundary~\cite{bekmambetova2016dissipative}
\begin{multline}
\dxhat \frac{\dyhat}{2} \left( \frac{\eps_x}{\dt} + \frac{\sigma_x}{2} \right) E_x|_{i+\half,1}^{n+1} = \\
\dxhat \frac{\dyhat}{2} \left( \frac{\eps_x}{\dt} - \frac{\sigma_x}{2} \right) E_x|_{i+\half,1}^{n} \\
+ \dxhat H_z|_{i+\half,1+\half}^{n+\half} - \dxhat H_z|_{i+\half,1}^{n+\half}\,.
\label{eq:updateExS}
\end{multline}
A similar equation can be derived to update the~$E_x$~nodes on the North boundary. The same approach is followed to write the update equations for the $E_y$ nodes.

\subsubsection{Update Equation for the $H_z$ Nodes}
Since all magnetic nodes fall strictly inside the fine region, we can use a standard FDTD update for all of them~\cite{Gedney}
\begin{multline}
	\dxhat \dyhat \frac{\mu}{\dt} H_z|_{i+\half, j+\half}^{n+\half} = 
			\dxhat \dyhat \frac{\mu}{\dt} H_z|_{i+\half, j+\half}^{n-\half} \\
			- \dxhat E_x|_{i+\half,j}^n + \dxhat E_x|_{i+\half,j+1}^n \\
			+ \dyhat E_y|_{i,j+\half}^n - \dyhat E_y|_{i+1,j+\half}^n ,
	\label{eq:updateH}
\end{multline}
 where $\mu$ denotes the average permittivity on the edge where $H_z|_{i+\half, j+\half}^{n+\half}$ is sampled.

\subsection{Dynamical System Formulation}

\added{As described in~\cite{bekmambetova2016dissipative}, update equations for all the electric and magnetic fields like~\eqref{eq:updateEx}-\eqref{eq:updateH} in the fine grid can be cast into the form of a discrete-time dynamical system}
\begin{subequations}
\begin{eqnarray}
	(\R + \F) \xhat^{n+1} & = &  (\R - \F) \xhat^{n} + \B\uhat^{n+\half}\,, \label{eq:sys1a} \\
	\yhat^n & = & \Lmatr^T \xhat^n\,.
	\label{eq:sys1b}
\end{eqnarray}
\end{subequations}

\added{The state vector~$\xhat^n$ consists of column vectors~$\Ehat_x^n$, $\Ehat_y^n$ and~$\Hzhat^{n-\half}$, which collect all $E_x$, $E_y$ and~$H_z$ samples in the fine region, respectively}
\begin{equation}
	\xhat^{n} =
	\begin{bmatrix}
		\Ehat_x^n \\
		\Ehat_y^n \\
		\Hhat_z^{n-\half}
	\end{bmatrix}\,.
	\label{eq:x}
\end{equation}

The input~$\uhat^{n+\half}$ and output $\yhat^{n}$ of~\eqref{eq:sys1a}-\eqref{eq:sys1b} contain all magnetic and electric fields on the North, South, West and East boundaries of the fine region
\begin{equation}
\uhat^{n+\half} = 
\begin{bmatrix}
\Hhat_S^{n+\half} \\
\Hhat_N^{n+\half} \\
\Hhat_W^{n+\half} \\
\Hhat_E^{n+\half} \\						
\end{bmatrix} \,,
\quad
\yhat^{n} = 
\begin{bmatrix}
\Ehat_S^{n} \\
\Ehat_N^{n} \\
\Ehat_W^{n} \\
\Ehat_E^{n} \\						
\end{bmatrix}\,.
\end{equation}
We can see that hanging variables are interpreted as the input of the FDTD model for the region, and the collocated electric fields are seen as the output. This input-output interpretation of the fine region equations is needed to apply MOR and, more importantly, to be able to reconnect the reduced model to the coarse grid without losing stability.
\added{Detailed expressions for the coefficient matrices in~\eqref{eq:sys1a}-\eqref{eq:sys1b} can be found in~\cite{bekmambetova2016dissipative}.}

\section{Model Order Reduction}

\label{sec:mor}

Because of grid refinement, the size of $\xhat^n$ can be quite large, and significantly increase CPU time. To mitigate this issue, system~\eqref{eq:sys1a}-\eqref{eq:sys1b} can be compressed using MOR. 
Among the many MOR methods available, we choose SPRIM~\cite{SPRIM}, since it will preserve the block structure of~\eqref{eq:sys1a}-\eqref{eq:sys1b}. Using  the robust Arnoldi iteration, SPRIM creates a projection matrix
\begin{eqnarray}
& \vect{V} = \begin{bmatrix}
\vect{V}_1 &  \vect{0} \\
\vect{0} & \vect{V}_2 \\
\end{bmatrix}
\label{eq:V}
\end{eqnarray}
suitable to approximate the full state vector $\xhat^n$ through a new state vector $\xt^n$ of smaller size
\begin{equation} \label{states}
\xhat^{n} \approx \vect{V}\tilde{\vect{x}}^{n} \,.
\end{equation}
Matrix $\vect{V}$ consists of two blocks $\vect{V}_1$~and~$\vect{V}_2$. The number of rows of $\matr{V}_1$ is equal to the number of electric field samples in $\Ehat_x^n$ and $\Ehat_y^n$. The number of rows in $\matr{V}_2$ is equal to the size of $\Hhat_z^{n-\half}$.
By substituting~\eqref{states} into~\eqref{eq:sys1a}-\eqref{eq:sys1b}, and multiplying the first equation by $\vect{V}^T$ on the left, we obtain the reduced model
\begin{subequations}
	\begin{eqnarray}
	(\Rt + \Ft) \tilde{\vect{x}}^{n+1} & = &  (\Rt - \Ft) \xt^{n} + \tilde{\matr{B}} \uhat^{n+\half}\,, \label{eq:rom1a} \\
	\yhat^n & = & \Lmatrt^T \xt^n\,,
	\label{eq:rom1b}
	\end{eqnarray}
\end{subequations}
where $\Rt = \vect{V}^T \hat{\vect{R}} \vect{V}$ can be written as
\begin{equation}
\Rt = \begin{bmatrix}
\tilde{\vect{R}}_{11}/\dt & -\half \tilde{\vect{K}} \\
-\half \tilde{\vect{K}}^T & \tilde{\vect{R}}_{22}/\dt \\
\end{bmatrix}\,,
\end{equation}
$\Ft = \vect{V}^T \hat{\vect{F}} \vect{V}$, $\Bt =\vect{V}^T \hat{\vect{B}}$, and $\Lmatrt =\vect{V}^T \hat{\vect{L}}$.

The size of reduced model~\eqref{eq:rom1a}-\eqref{eq:rom1b} is typically much lower than the size of~\eqref{eq:sys1a}-\eqref{eq:sys1b}. The size and accuracy of the reduced model can be controlled through the choice of the number of columns in the projection matrix $\matr{V}$.

\section{Incorporation of the Reduced Model into the Main Coarse Grid}

\label{sec:connection}

The goal of this section is to \added{propose a stable way} to couple the reduced model to the coarse grid, while maintaining stability. An improper connection between different domains is indeed a common source of instability. Since the fields on the reduced model boundaries are sampled with a fine resolution, different from the resolution of the coarse grid, a suitable interpolation rule will be introduced. For coupling the reduced model to the coarse grid, we generalize the approach of~\cite{bekmambetova2016dissipative} to handle reduced models. To realize the coupling, we will combine three sets of equations: the update equations for the reduced model~\eqref{eq:rom1a}-\eqref{eq:rom1b}, the update equation for the coarse grid nodes at the interface with the reduced model, and an interpolation rule. 

\subsection{State Equations for the Coarse Fields at the Interface}

We consider the interface shown in Fig.~\ref{fig:connectionA} between the North boundary of the coarse grid, and the South boundary of the reduced model. The other edges can be treated similarly.
For the $E_x$ nodes of the coarse grid that fall at the interface with the reduced model, we introduce hanging variables to facilitate the connection of the reduced model. The hanging variable is the magnetic field collocated with the electric field on the boundary. For example, in the case in Fig.~\ref{fig:connectionA}, the hanging variable is $H_z|_{i+\half, j}$. 
Using this variable, a modified update equation analogous to~\eqref{eq:updateExS} is written for $E_x|_{i+\half, j}$
\begin{multline}
\dx\ESleft E_x|_{i+\half,j}^{n+1} = \label{eq:updateExboundary}
\dx\ESright E_x|_{i+\half,j}^{n} \\
+ \dx H_z|_{i+\half, j}^{n+\half}
- \dx H_z|_{i+\half,j-\half}^{n+\half}
\,.
\end{multline}
In~\eqref{eq:updateExboundary}, $\eps_x$ and $\sigma_x$ are the permittivity and conductivity in the half coarse cell below the interface.
The hanging variable~$H_z|_{i+\half,j}^{n+\half}$ can be considered as the input for the coarse grid. The use of the hanging variable~$H_z|_{i+\half,j}^{n+\half}$~is necessary since a standard FDTD equation for $E_x|_{i+\half,j}^{n+1}$ would involve a magnetic field sample inside the reduced model, which is not available. We will also see in the next section that, through hanging variables, we can systematically ensure stability.

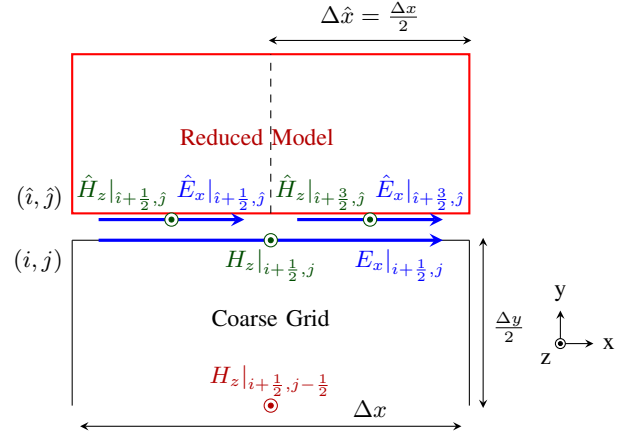
\begin{figure}[t]
	
	\centering
	\begin{tikzpicture}
	[
	scale=1.1,
	arr/.style={very thick,->,Ecol,>=stealth,shorten >=10pt, shorten <=10pt},
	arrb/.style={very thick,->,Ecol,>=stealth,shorten >=10pt, shorten <=10pt},
	labelE/.style={font=\fontsize{\sizefont}{\sizefont}\selectfont, Ecol},
	labelH/.style={font=\fontsize{\sizefont}{\sizefont}\selectfont, Hcol},
	labelHV/.style={font=\fontsize{\sizefont}{\sizefont}\selectfont, HV},
	label/.style={font=\fontsize{\sizefont}{\sizefont}\selectfont}
	]

	\pgfmathsetmacro{\sizefont} {9}
	
	\pgfmathsetmacro{\W}{2.4}
	\pgfmathsetmacro{\H}{2}
	\pgfmathsetmacro{\dH}{0}
	\pgfmathsetmacro{\dHb}{0}
	\pgfmathsetmacro{\Rbig}{0.08}
	\pgfmathsetmacro{\Rsmall}{0.03}
	\pgfmathsetmacro{\extraW}{\H/4}
	\pgfmathsetmacro{\extraH}{\H/4}
	\pgfmathsetmacro{\axisOffsetx}{\W/4}
	\pgfmathsetmacro{\axisOffsety}{-0.75*\H}
	
	\pgfmathsetmacro{\gap}{0.25};

	\draw (0,-\gap) -- (2*\W,-\gap);
	\draw (0,-\gap) -- (0, -1*\H-\gap);
	\draw (2*\W, -\gap) -- (2*\W, -1*\H-\gap);
	\draw [red,thick] (0,0.3*\gap) rectangle (2*\W, \H);
	\draw [labelH] (0.93*\W, 0.5*\H ) node {Reduced Model};
	
	\draw [Hcol, fill=white] (1*\W, -1*\H-\gap) circle (\Rbig);
	
	\draw [fill, Hcol] (1*\W, -1*\H-\gap) circle (\Rsmall);

	\draw[black,dashed] (\W,0.3*\gap) -- (\W,\H);

	
	\draw[labelHV] (0.25*\W,\dH) node [above] {$\hat{H}_z|_{\ihat+\half, \jhat}$};
	\draw[labelHV] (1.25*\W,\dH) node [above] {$\hat{H}_z|_{\ihat+\frac{3}{2}, \jhat}$};
	
	\draw[labelH] (1*\W,-1*\H-\gap) node [above] {$H_z|_{i+\half, j-\half}$};
	\draw[labelHV] (1*\W,-\dHb-\gap) node [below] {$H_z|_{i+\half, j}$};
	
	\draw[labelE] (0.75*\W,\dH) node [above] {$\hat{E}_x|_{\ihat+\half, \jhat}$};
	\draw[labelE] (1.75*\W,\dH) node [above] {$\hat{E}_x|_{\ihat+\frac{3}{2}, \jhat}$};
	\draw[labelE] (1.65*\W,-\dHb-\gap) node [below] {$E_x|_{i+\half, j}$};
	
	\draw[arrb] (0,-\dHb-\gap) -- (2*\W,-\dHb-\gap);
	\draw[arr] (0,\dH) -- (\W,\dH);
	\draw[arr] (\W,\dH) -- (2*\W,\dH);
	
	\draw [HV, fill=white] (0.5*\W, \dH) circle (\Rbig);
	\draw [HV, fill=white] (1.5*\W, \dH) circle (\Rbig);
	\draw [HV, fill=white] (1*\W, -\dHb-\gap) circle (\Rbig);
	
	\draw [fill, HV] (0.5*\W, \dH) circle (\Rsmall);
	\draw [fill, HV] (1.5*\W, \dH) circle (\Rsmall);
	\draw [fill, HV] (1*\W, -\dHb-\gap) circle (\Rsmall);

	\draw[label] (0,0) node [above left] {$(\ihat,\jhat)$};
	\draw[label] (0,-\gap) node [below left] {$(i,j)$};
	
	\draw[label] (\W,-0.7*\H) node [above] {Coarse Grid};
	
	\draw[->,black,>=stealth] (2*\W+\extraH+\axisOffsetx, 0+\axisOffsety) -- (2*\W+\extraH+\axisOffsetx, 0.2*\H+\axisOffsety);
	\draw[->,black,>=stealth] (2*\W+\extraH+\axisOffsetx, 0+\axisOffsety) -- (2*\W+\extraH+\axisOffsetx + 0.2*\H, 0+\axisOffsety);
	\draw [fill=white] (2*\W+\extraH+\axisOffsetx, \axisOffsety) circle (0.025*\W);
	\draw [fill=black] (2*\W+\extraH+\axisOffsetx, \axisOffsety) circle (0.01*\W);
	\draw[label] (2*\W+\extraH+\axisOffsetx + 0.2*\H, 0+\axisOffsety) node [right] {x};
	\draw[label] (2*\W+\extraH+\axisOffsetx, 0+\axisOffsety+ 0.2*\H) node [above] {y};
	\draw[label] (2*\W+\extraH+\axisOffsetx, 0+\axisOffsety) node [below left] {z};
	
	\draw[<->,black,>=stealth, shorten >=3pt, shorten <=3pt] (0.95*\W,\H+\extraH/3) -- (2.05*\W, \H+\extraH/3);	
	\draw[label] (1.5*\W, \H+\extraH/3) node [above] {$\dxhat=\frac{\dx}{2}$};
	\draw[<->,black,>=stealth, shorten >=3pt, shorten <=3pt] (2*\W+\extraW/3,-\H-1.5*\gap) -- (2*\W+\extraW/3, -0.5*\gap);	
	\draw[label] (2*\W+\extraW/3, -1.3*\H/2) node [right] {$\frac{\dy}{2}$};	

	\draw[<->,black,>=stealth, shorten >=3pt, shorten <=3pt] (0,-\H-\extraH/1) -- (2*\W, -\H-\extraH/1);	
	\draw[label] (1.5*\W, -\H-\extraH/1) node [above] {$\dx$};
	
	\end{tikzpicture}
	\caption{Connection scenario considered in Sec.~\ref{sec:connection} for the case of $r=2$. A virtual gap has been inserted between the two subsystems for clarity.}
	\label{fig:connectionA}
\end{figure}

For all edges at the boundary of the coarse grid, a hanging variable and an equation analogous to~\eqref{eq:updateExboundary} is defined. All these update equations can be written in matrix form as

\begin{multline}
\Dl \Dlp \left( \frac{\Deb}{\dt} + \frac{\Dseb}{2} \right) \matr{y}^{n+1} = \Dl \Dlp \left( \frac{\Deb}{\dt} - \frac{\Dseb}{2} \right) \matr{y}^{n} + \\ \Dl \matr{G}^T \matr{H}^{n+\half} - \Dl \matr{G}^T \matr{U}^{n+\half}\,,
\label{eq:updateEBC}
\end{multline}
where $\vect{y}^{n}$ collects all the coarsely-sampled electric fields at the interfaces, and $\Deb$ and $\Dseb$ are the diagonal matrices containing the permittivity and conductivity values in the half cells on the coarse grid boundary. Matrix $\Dl$ is diagonal and contains the length of the coarse grid on the interface. Diagonal matrix $\Dlp$ contains the length of the half edges of the coarse grid that intersect the interface. Matrix $\matr{G}$ is a discrete derivative operator and contains only $-1$, $+1$ and $0$ entries. Vector $\matr{H}^{n+\half}$ contains the magnetic fields in the coarse grid near the interface with the reduced model. Input vector $\vect{U}^{n+\half}$ collects all the hanging variables on the coarse side of the interface, such as $H_z|_{i+\half,j}^{n+\half}$ in Fig.~\ref{fig:connectionA}.
\subsection{Interpolation Rule}

The fields at the boundaries of the coarse mesh and of the reduced model need to be linked by a suitable interpolation rule. This is necessary to satisfy boundary conditions, and because they are sampled with different resolution. We use the interpolation rule in~\cite{bekmambetova2016dissipative}, which will ensure stability, as we shall see in the next section. 
Between the boundary electric fields, we impose the following relation
\begin{equation}
\yhat^n = \T \vect{y}^n \qquad \forall n\,,
\label{eq:equalE}
\end{equation}
where $\T$ is the transformation matrix that sets adjacent coarse and fine grid electric fields to be equal. As discussed in~\cite{bekmambetova2016dissipative}, a reciprocal constraint must be imposed on the magnetic fields on the boundary in order to ensure stability
\begin{equation}
\matr{U}^{n+\half} = \frac{1}{r}\T^T \uhat^{n+\half} \qquad \forall n\,.
\label{eq:averH}
\end{equation}

\subsection{Coupling of Reduced Model and Coarse Grid}
The final update equations for the reduced model are derived by combining the reduced model~\eqref{eq:rom1a}-\eqref{eq:rom1b}, update equation~\eqref{eq:updateEBC} for the boundary fields of the coarse grid, and interpolation rules~\eqref{eq:equalE}~and~\eqref{eq:averH}. First, we substitute~\eqref{eq:averH}~into~\eqref{eq:updateEBC}
\begin{multline}
\Dl \Dlp \left( \frac{\Deb}{\dt} + \frac{\Dseb}{2} \right) \matr{y}^{n+1} =  \Dl\Dlp \left( \frac{\Deb}{\dt} - \frac{\Dseb}{2} \right) \matr{y}^{n} + \\
\Dl \matr{G}^T \matr{H}^{n+\half} - \frac{1}{r} \Dl \matr{G}^T\T^T \uhat^{n+\half}\,.
\label{eq:updateEBC2}
\end{multline}
Then, equation~\eqref{eq:equalE}~is substituted into~\eqref{eq:rom1b}
\begin{equation}
\T\vect{y}^{n+1} = \Lmatrt^T\xt^{n+1} \,.
\label{eq:yandxt}
\end{equation}
Equations~\eqref{eq:rom1a}, \eqref{eq:updateEBC2}, and~\eqref{eq:yandxt} are combined in a linear system
\begin{equation}
\vect{A}_1 \vect{z}^{n+1} = \vect{A}_2 \vect{z}^{n} + \Dl \matr{G}^T\begin{bmatrix}
\matr{H}^{n+\half}\\
\matr{0}\\
\matr{0}\\
\end{bmatrix} \,,
\label{eq:sysform}
\end{equation}
where
\begin{equation}
\vect{z}^n = \begin{bmatrix}
\matr{y}^{n}\\
\uhat^{n-\half} \\
\xt^{n} \\
\end{bmatrix} \,,
\end{equation}

\begin{equation}
\vect{A}_1 = \begin{bmatrix}
\Dl \Dlp \left( \frac{\Deb}{\dt} + \frac{\Dseb}{2} \right) & \frac{1}{r}\Dl\matr{G}^T\matr{T}^T & \vect{0}\\
\vect{T}  &  \vect{0} &-\Lmatrt^T \\
\vect{0} &  -\Bt & (\Rt + \Ft)\\

\end{bmatrix} \,,
\end{equation}
\begin{equation}
\vect{A}_2 = \begin{bmatrix}
\Dl\Dlp \left( \frac{\Deb}{\dt} - \frac{\Dseb}{2} \right) & \matr{0} &  \matr{0} \\
\matr{0} &\matr{0}& \matr{0}\\
\matr{0} & \matr{0} & (\Rt-\Ft)\\

\end{bmatrix} \,.
\end{equation}
Multiplying~\eqref{eq:sysform} on the left by $\vect{A}_1^{-1}$, we obtain
\begin{equation}
\vect{z}^{n+1} = \matr{A}_1^{-1}\matr{A}_2\vect{z}^{n} + \matr{A}_1^{-1}\Dl\matr{G}^T \begin{bmatrix}
\matr{H}^{n+\half}\\
\matr{0}\\
\matr{0}\\
\end{bmatrix} \,.
\label{eq:connExplicit}
\end{equation}
Equation~\eqref{eq:connExplicit} is explicit, and is used to update the state of the reduced model~$\xt^n$, and the electric fields~$\vect{y}^n$ at the interface with the coarse grid. The coefficient matrices in~\eqref{eq:connExplicit} can be precomputed once before runtime. A sparse LU factorization can be performed on $\vect{A}_1$ to efficiently compute its inverse multiplied by $\vect{A}_2$ and $\Dl\vect{G}^T$. 
Ultimately, the proposed scheme consists of conventional FDTD equations to update the coarse grid, and~\eqref{eq:connExplicit} to update the reduced model state and the interface. If multiple reduced models are present, equation~\eqref{eq:connExplicit} is applied to each one of them.

\section{Stability Analysis and CFL Limit Extension}
\label{sec:cflanalysis}

In this section, we analyse the stability of the FDTD scheme with embedded reduced models. We derive the CFL limit of the proposed scheme, and show how it can be extended \added{beyond the fine grid's CFL limit}. We advocate the stability theory proposed in~\cite{bekmambetova2016dissipative,cnf-2016-epeps-fdtdpassivity}, which is based on the concept of energy dissipation. This theory is particularly convenient to investigate the stability of  advanced FDTD setups where different meshes and models are coupled together. 

\subsection{Passivity of the Coarse Grid and the Interpolation Rule}

From the perspective of system theory, the obtained scheme can be seen as the connection of three subsystems: the coarse mesh, the interpolation rule, and the reduced model. If all subsystems are passive, i.e. unable to generate energy on their own, the final scheme will be passive by construction, and thus stable~\cite{bekmambetova2016dissipative,jnl-2007-tadvp-fundamentals}. As proved in~\cite{bekmambetova2016dissipative}, the coarse mesh is passive for any time step satisfying its own CFL limit. Interpolation rule~\eqref{eq:equalE}-\eqref{eq:averH} can be shown to be lossless for any time step~\cite{bekmambetova2016dissipative}. In order to establish the stability of the overall scheme, we have to investigate when the reduced model is passive.

\subsection{Passivity Conditions for the Fine Grid}

The reduced model has been derived from the fine grid model~\eqref{eq:sys1a}-\eqref{eq:sys1b}. In order to be passive, this model must satisfy the following three passivity conditions~\cite{bekmambetova2016dissipative}
\begin{subequations}
	\begin{eqnarray}
	& \R = \R^T >  0\,,  \label{eq:cond1} \\
	& \F + \F^T \ge 0\,, \label{eq:cond2}  \\
	& \B =  \Lmatr \Shat \,, \label{eq:cond3} 
	\end{eqnarray}
\end{subequations}
where
\begin{equation}
\Shat = \Lmatr^T \B = 
\begin{bmatrix}
- \dxhat \matr{I}_{\hat{N}_x}	& \matr{0} 					& \matr{0}  				& \matr{0} \\
\matr{0} 					&  \!\!\!\dxhat \matr{I}_{\hat{N}_x}		& \matr{0}  				& \matr{0} \\
\matr{0}  					& \matr{0} 					& \!\!\!\dyhat \matr{I}_{\hat{N}_y}		& \matr{0} 			\\
\matr{0} 					& \matr{0}  				& \matr{0} 					& \!\!\!- \dyhat \matr{I}_{\hat{N}_y}	\\
\end{bmatrix}\,.
\label{eq:S}
\end{equation}
Condition~\eqref{eq:cond1} can be shown to be a generalization of the CFL limit for a region with inhomogeneous materials~\cite{bekmambetova2016dissipative}. Condition~\eqref{eq:cond2} is satisfied if all conductivities $\sigma_x$ and $\sigma_y$ in the fine region are non-negative, which is usually the case. It can be shown, by direct substitution, that~\eqref{eq:cond3} always holds.  In conclusion, the original model for the fine grid~\eqref{eq:sys1a}-\eqref{eq:sys1b} is passive under the CFL limit of the fine grid~\cite{bekmambetova2016dissipative}.

\subsection{Passivity of the Reduced Model}

We now investigate the passivity of the reduced model~\eqref{eq:rom1a}-\eqref{eq:rom1b} that has been produced for the fine grid. The reduced model will be passive if
\begin{subequations}
	\begin{eqnarray}
	& \Rt = \Rt^T >  0\,,  \label{eq:cond1rom} \\
	& \Ft + \Ft^T \ge 0\,, \label{eq:cond2rom}  \\
	& \Bt =  \Lmatrt \Shat \,. \label{eq:cond3rom} 
	\end{eqnarray}
\end{subequations}
The following theorem shows that \added{the proposed reduction process preserves the passivity of the fine grid model by construction, which is a new result}.
\begin{theorem} \label{thm:passivityrom}
Reduced model~\eqref{eq:rom1a}-\eqref{eq:rom1b} is passive for any $\dt$ that satisfies the CFL limit~\eqref{eq:cond1} of the fine grid.
\end{theorem}
\begin{proof}
Since $\matr{V}$ is full rank by construction, and $\Rt = \matr{V}^T \R \matr{V}$, condition~\eqref{eq:cond1} implies~\eqref{eq:cond1rom}.  In the same way, since $\Ft = \matr{V}^T \F \matr{V}$, condition~\eqref{eq:cond2} implies~\eqref{eq:cond2rom}. Finally, by multiplying~\eqref{eq:cond3} on the left by $\matr{V}^T$, we obtain~\eqref{eq:cond3rom}.
\end{proof}

\subsection{Stability of the FDTD Scheme with Reduced Models}

To summarize, we have shown that:
\begin{itemize}
	\item the interpolation rule is lossless, and thus passive, for any $\dt$;
	\item the coarse grid is passive for any $\dt$ under its own CFL limit;
	\item the reduced model of the fine grid is passive for any $\dt$ under the CFL limit of the fine grid.
\end{itemize}

The overall scheme, which results from the connection of these three subsystems, \added{will be therefore passive and stable \emph{by construction} up to the CFL limit of the fine grid because of the following theorem.}
\begin{theorem} \label{thm:passivitycascade}
\added{The system obtained by connecting an arbitrary number of passive systems is passive as well, and does not have any poles outside the unit circle of the complex plane.}
\end{theorem}
\begin{proof}
	\added{See Appendix~\ref{sec:append}}.
\end{proof}
\added{We have provided a formal proof of stability of the proposed FDTD scheme with reduced models. The proof is valid for an arbitrary number of embedded reduced models, since one must simply require that all reduced models are passive. This is a new result, since earlier works on FDTD schemes with embedded reduced models did not provide a rigorous stability proof~\cite{cangellaris1999rapid,denecker,kulas2003reduced}, or neglected losses and gave stability conditions that become increasingly complicated as the number of reduced models grows~\cite{kulas2004stability}. Stability is essential for the adoption of novel FDTD schemes by commercial solvers and industry, where a rigorous guarantee of stability is mandatory. Another important feature of the proposed method is that its CFL limit is known a priori, without requiring the calculation of expensive norms~\cite{kulas2004stability}.}

\subsection{Extension of the CFL Limit}

\added{When $\dt$ exceeds the CFL limit of the finest grid, the corresponding reduced model will not be passive anymore}, and will destabilize the whole scheme. This constraint can be relaxed by extending the CFL limit of the reduced models prior to their connection to the main grid. The extension can be achieved with a simple perturbation of the model coefficients. For a time step beyond the CFL limit of a given reduced model, conditions~\eqref{eq:cond2rom} and~\eqref{eq:cond3rom} will still hold, since they are independent from $\dt$. 
Condition~\eqref{eq:cond1rom} will be violated, but can be restored with a perturbation of $\tilde{\vect{K}}$. Using the Schur's compliment~\cite{Gol96}, one can show that~\eqref{eq:cond1rom} is equivalent to~\cite{jnl-2014-tmtt-fdtd}
\begin{equation}
 s_k \leq \frac{2}{\dt}  \quad \forall k \,.
\label{eq:CFLgeneralized} 
\end{equation}
where $s_k$ are the singular values of $\tilde{\vect{R}}_{11}^{-\half}\tilde{\vect{K}}\tilde{\vect{R}}_{22}^{-\half}$. Above the CFL limit of the reduced model, some singular values will exceed the threshold in~\eqref{eq:CFLgeneralized}. By perturbing $\tilde{\vect{K}}$ with the method in~\cite{jnl-2014-tmtt-fdtd}, condition~\eqref{eq:CFLgeneralized} can be enforced. \added{Numerical tests in Sec.~\ref{sec:numerical_examples} will show that an extension of the CFL limit by 2X or 3X can be typically achieved with minimal impact on accuracy. This enforcement procedure enables the use of a larger time step in the \textit{whole} domain, including the entire coarse grid, and results in a significant saving of CPU time. An interesting aspect is that the CFL limit is enhanced by acting only on the reduced models of the refined grids. This is a novel result compared to~\cite{jnl-2014-tmtt-fdtd}, where MOR and singular value perturbation had to be applied to the equations of the whole system. This feature makes the proposed method more efficient and scalable. Finally, it is remarkable that the proposed scheme achieves stability above the CFL limit of the fine grid without resorting to implicit schemes. Only explicit operations are required at runtime, leading to higher efficiency over implicit alternatives, as numerical results will show.}

\section{Numerical Examples}
\label{sec:numerical_examples}

Several test cases are provided to numerically validate the stability and performance of the proposed method, which was implemented in Matlab with vectorized operations for maximum efficiency. Simulations were run on a computer with a 3.6~GHz CPU and  8~GB of memory.

\subsection{2D Cavity}
\label{sec:ne_stability}

\begin{figure}[t]
	\centering
	\begin{tikzpicture}
	[scale = 0.85,
	arr/.style={->,>=stealth,shorten >=1pt},
	label/.style={font=\fontsize{\sizefont}{\sizefont}\selectfont}
	]
	
	\pgfmathtruncatemacro{\GR}{5};
	\pgfmathtruncatemacro{\dxCmm}{2};
	\pgfmathtruncatemacro{\dyCmm}{2};
	\pgfmathsetmacro{\lenx}{4.5};
	\pgfmathsetmacro{\leny}{4.5};
	\pgfmathsetmacro{\subgrAx}{0.8};
	\pgfmathsetmacro{\subgrAy}{0.8};
	\pgfmathsetmacro{\subgrBx}{1.8};
	\pgfmathsetmacro{\subgrBy}{1.8};
	\pgfmathsetmacro{\curx}{0.4};
	\pgfmathsetmacro{\cury}{0.4};
	\pgfmathsetmacro{\probex}{4.2};
	\pgfmathsetmacro{\probey}{4.2};
	\pgfmathtruncatemacro{\truelenxmm}{1};
	\pgfmathtruncatemacro{\truelenymm}{1};
	\pgfmathtruncatemacro{\subgrlenxmm}{20};
	\pgfmathtruncatemacro{\subgrlenymm}{20};
	
	\pgfmathsetmacro{\Rbig}{0.035}

	\pgfmathsetmacro{\sizefont} {9}
	\draw [very thick] (0, 0) rectangle (\lenx,\leny);
	\draw [->, >=stealth, shorten >=2pt, shorten <=7pt] (\lenx/2-0.5, \leny+0.4) -- (\lenx/2, \leny);
	\draw[label] (\lenx/2-0.5, \leny+0.4) node {PEC};
	\draw [red, fill = red] (\curx,\cury) circle (\Rbig);
	\draw[font=\fontsize{\sizefont}{\sizefont}\selectfont] (\curx,\cury) node [right] {Source};	
	\draw [blue, fill = blue] (\probex,\probey) circle (\Rbig);
	\draw[label] (\probex,\probey) node [left] {Probe};	
	
	\draw [dashed, black, thick] (\subgrAx, \subgrAy) rectangle(\subgrBx, \subgrBy);

	\draw[label] (\lenx*0.7,\leny*0.6) node {\begin{tabular}{c} $\dxhat$ = $\dx/$\GR \\ $\dyhat$ = $\dy/$\GR \end{tabular}};
	\draw[label] (0,\leny) node [below right] {\begin{tabular}{c} $\dx$ = \dxCmm\ cm \\ $\dy$ = \dyCmm\ cm \end{tabular}};
	
	\pgfmathsetmacro{\axisoffs}{\lenx/6}
	\draw[->,black,>=stealth] (\lenx + \axisoffs,\leny/2) -- (\lenx + \lenx/10+\axisoffs,\leny/2);
	\draw[->,black,>=stealth] (\lenx + \axisoffs,\leny/2) -- (\lenx + \axisoffs,\leny/2 + \lenx/10);
	\draw [fill=white] (\lenx + \axisoffs,\leny/2) circle (0.02*\lenx);
	\draw [fill=black] (\lenx + \axisoffs,\leny/2) circle (0.0075*\lenx);
	\draw[label] (\lenx + \axisoffs+\lenx/10,\leny/2) node [right] {x};
	\draw[label] (\lenx + \axisoffs,\leny/2 + \lenx/10) node [above] {y};
	\draw[label] (\lenx + \axisoffs,\leny/2) node [below left] {z};
	
	\draw[<->,black,>=stealth] (0,-0.15) -- (\lenx,-0.15);
	\draw (\lenx/2,-0.15) node [below] {\truelenxmm\ m};
	\draw[<->,black,>=stealth] (-0.15,0) -- (-0.15,\leny);
	\draw (-0.15,\leny/2) node [above, rotate = 90] {\truelenymm\ m};
	
	\draw[<->,black,>=stealth] (\subgrAx,\subgrBy+0.1) -- (\subgrBx,\subgrBy+0.1);
	\draw [label](\subgrAx/2+\subgrBx/2+0.1,\subgrBy+0.1) node [above] {\subgrlenxmm\ cm};
	
	\draw[<->,black,>=stealth] (\subgrAx-0.1,\subgrAy) -- (\subgrAx-0.1,\subgrBy);
	\draw[label] (\subgrAx-0.1,\subgrBy/2+\subgrAy/2) node [above, rotate=90] {\subgrlenymm\ cm};
	
	\end{tikzpicture}
	\caption{Layout of the PEC cavity considered in Sec.~\ref{sec:ne_stability}. The dashed box denotes the refinement region.}
	\label{fig:layout_stability}
\end{figure}
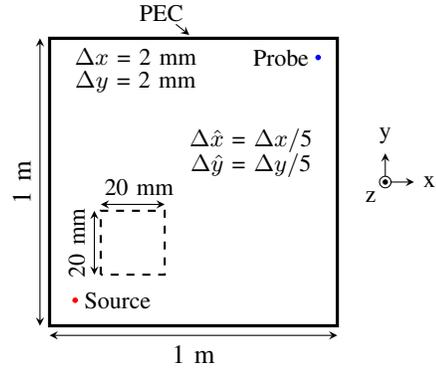
In this section, the proposed method is applied to the empty 1~m~$\times$~1~m 2D cavity in Fig.~\ref{fig:layout_stability}. The cavity has perfect electric conductor (PEC) walls. The whole region is discretized with a coarse mesh with $\dx$~=~$\dy$~=~2 cm, except for a central 0.2~m~$\times$~0.2~m area, discretized with a fine mesh ($r= 5$). The purpose of this test is to verify the stability of the proposed method. The cavity is excited by a Gaussian pulse source with bandwidth of 0.5 GHz. 

Using the proposed technique, a reduced model was created from the FDTD equations of the fine region~\eqref{eq:sys1a}-\eqref{eq:sys1b}. MOR reduced their order from 7,600 to 1,200. The CFL limit of the reduced model was extended by 2X, prior to its embedding into the coarse mesh. The proposed method was run for $10^6$ time steps. The computed magnetic field at the probe, shown in Fig.~\ref{fig:result_stability}, confirms the stability of the proposed method. 

Table~\ref{Cavity_table} gives the CPU time taken by the proposed method, by standard FDTD \added{(Yee scheme)} with an all-coarse mesh, by standard FDTD with an all-fine mesh,  by a subgridding method~\cite{bekmambetova2016dissipative} \added{and by a hybrid one-step leapfrog ADI-FDTD subgridding method~\cite{YangADIFDTDsub2D}}. The proposed method was also tested with MOR enabled but no CFL limit extension, and vice versa. All explicit methods were run at a time step of 99\% of their CFL limit. \added{The proposed method and ADI-FDTD subgridding were run at a time step two times larger than the CFL limit of the fine grid.} In terms of accuracy, all methods are in very good agreement, as shown in Fig.~\ref{fig:Freq_cavity}. Since, in this case, the fine region does not contain any complex object, all methods are expected to give comparable accuracy. In terms of execution time, \added{Table~\ref{Cavity_table} shows that the} proposed method provides a 7.7X speed-up with respect to a complete refinement of the FDTD mesh, a 60\% gain versus subgridding, \added{and is about twice faster than the ADI-FDTD subgridding scheme in~\cite{YangADIFDTDsub2D}}. We can see that both MOR and the extension of the CFL limit of the reduced model contribute to the computational gains of the proposed method.

\begin{figure}[t]
	\centering
	\includegraphics[width=0.95\columnwidth]{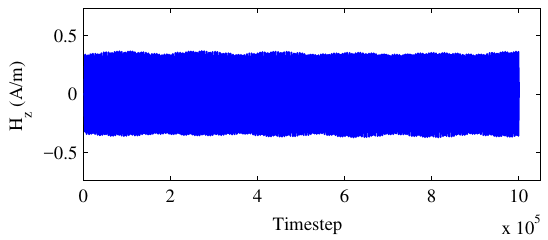}
	\caption{Magnetic field at the probe for the empty cavity of Sec.~\ref{sec:ne_stability}, computed with the proposed method for 10$^6$ time steps.}
	\label{fig:result_stability}
\end{figure}

\begin{table}[t]\centering 
	\centering
	\caption{Execution time for the 2D cavity in Sec.\ref{sec:ne_stability}. The MOR time is included in the runtime. The CFL number is defined as the ratio of the time step $\dt$ used in the simulation and the CFL limit.}
		\begin{tabular}{  l c  r  r }
			\hline
			Method                                 & CFL number  & Runtime   & Speed-up \\ \hline \hline
			\added{Yee scheme}, all-coarse grid                          & $0.99$  &  9.0 s     & - \\ \hline
			\added{Yee scheme}, all-fine grid						      & $0.99$  &  451.2 s    & - \\ \hline	
			Subgridding			                  & $0.99$  &  95.3 s    & 4.7X\\ \hline
			\added{ADI-FDTD Subgridding~\cite{YangADIFDTDsub2D}}			                  & $1.98$  &  146.7 s    & 3.1X\\ \hline
			Proposed (MOR  only)                    & $0.99$  &  131.1 s   & 3.4X\\ \hline 
			Proposed (CFL extension	only)	       & $1.98$  & 340.1 s   & 1.3X\\ \hline 
			Proposed			              & $1.98$  &  58.6 s & 7.7X\\ \hline
		\end{tabular}
	\label{Cavity_table}
\end{table}

\begin{figure}[t]\centering 

%
%
%
%
%
    \centering
    \includegraphics[width=0.95\columnwidth]{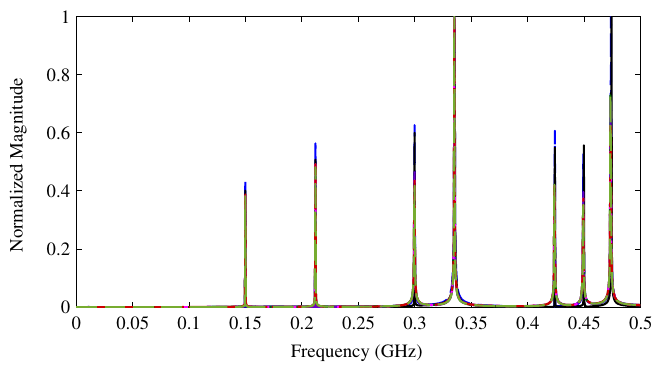}
	
	\caption{Frequency response of the cavity of Sec.~\ref{sec:ne_stability}, obtained with the \added{Yee scheme}  with an all-coarse mesh~(\color{legendBlue} ~\st{$\;\;$}~\st{$\;\;\;$}~\st{$\;\;$}~\color{black}), the \added{Yee scheme}  with an all-fine mesh (\color{black} ~\st{$\;\;\;\;\;\;\;$}~\color{black}), subgridding (\color{legendGreen}~\st{$\;\;\;\;\;\;$}~\color{black}), \added{ADI-FDTD subgridding~\cite{YangADIFDTDsub2D}} (\color{legendLightgreen}~\st{$\;\;\;$}~\st{$\;\;\;$}~\color{black}), the proposed method with only MOR (\color{legendPurple}~\st{$\;\;\;$}~\st{$\;\;\;$}~\color{black}), the proposed method with only CFL extension (\color{legendPink}~\st{$\;$}~\st{$\;$}~\st{$\;$}~\st{$\;$}~\color{black}), and the proposed method with both MOR and CFL extension~(\color{legendRed}~\st{$\;$}~\st{$\;\;\;$}~\st{$\;$}~\color{black}).}
	\label{fig:Freq_cavity}
\end{figure}

\subsection{Waveguide with Irises}
\label{sec:ne_filter}

We consider the 4~m $\times$ 0.7~m waveguide shown in Fig.~\ref{fig:filter}. The waveguide includes two thin PEC irises~\cite{cnf-2015-nemo-fdtdmor}. A coarse mesh is used in most of the domain, with resolution $\dx = \dy =$ 0.05~m. The two areas around the irises were instead meshed with a refined grid ($r=3$) in order to capture the irregular fields caused by the discontinuity. The refined regions are 0.5~m $\times$ 0.5~m wide, and are shown in Fig.~\ref{fig:filter}. In the proposed method, two reduced models were generated for these regions and embedded into the coarse grid. The ends of the waveguide are terminated with a perfectly matched layer (PML) which is 10 coarse cells deep. The waveguide is excited by a Gaussian pulse with 0.4~GHz bandwidth. A probe is located at the other end. All methods were run at 0.99 times their CFL limit. In the proposed method, the CFL limit was extended by two times, and MOR reduced the size of each fine grid model from 2,760 to 648. 

Fig.~\ref{fig:filter_results} compares the frequency response obtained with \added{the Yee scheme}, subgridding~\cite{bekmambetova2016dissipative} and the proposed method. We can clearly see that a coarse mesh does not accurately capture the frequency response of the system, due to its inability to properly resolve the irregular fields around the apertures. The other methods, which utilize a fine grid in the two critical regions, are in very good agreement with an all-fine FDTD simulation. In terms of runtime, Table~\ref{Filter_table} shows that subgridding \added{and ADI-FDTD subgridding~\cite{YangADIFDTDsub2D} are slightly faster than an all-fine FDTD simulation}, and that the proposed method provides the highest speed up among the compared techniques. For this case, the speed-up is more modest than in the previous example, due to the presence of two refined regions and to their size relative to the overall problem. 

\begin{figure}[t]
	\centering
	\includegraphics[scale=0.75]{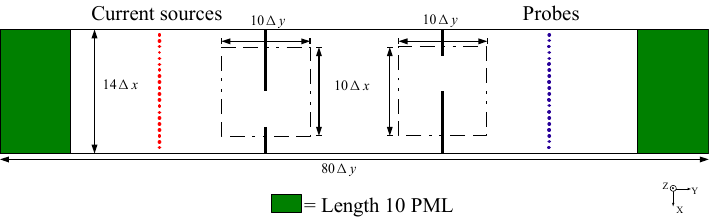}
	\caption{Geometry of the waveguide with irises considered in Sec.~\ref{sec:ne_filter}. The boxes indicate the regions replaced by a reduced model.}
	\label{fig:filter}
\end{figure}

\begin{figure}[t]\centering 
	\includegraphics[scale=0.8]{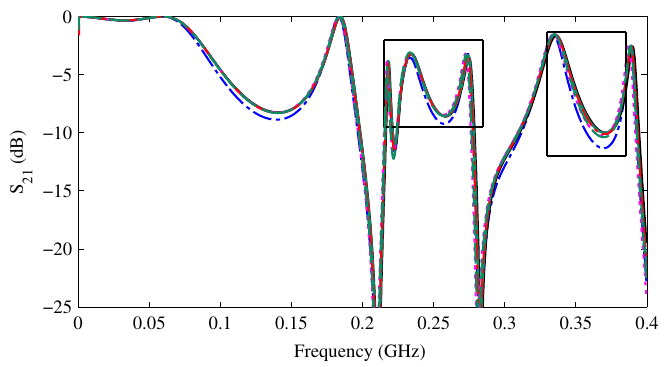}

	\begin{tikzpicture}
	\pgfmathsetmacro{\fraccolwid}{0.5};
	\pgfmathsetmacro{\dist}{4.5};
	
	\node[inner sep=0pt] (whitehead) at (0,0) [below] {\includegraphics[width = \fraccolwid\columnwidth]{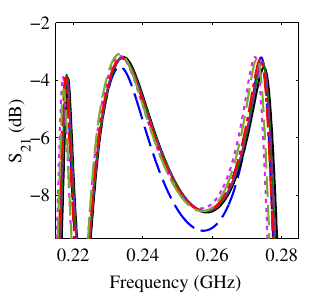}};
	
	\node[inner sep=0pt] (whitehead) at (\dist,0)[below] {\includegraphics[width = \fraccolwid\columnwidth]{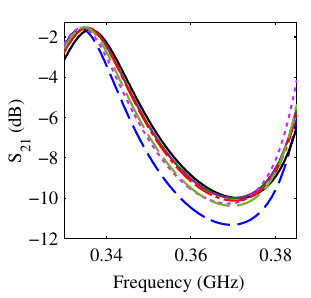}};

	\end{tikzpicture}

\caption{Frequency response of the waveguide of Sec.~\ref{sec:ne_filter}, obtained with \added{the Yee scheme} with an all-coarse mesh~(\color{legendBlue} ~\st{$\;\;$}~\st{$\;\;\;$}~\st{$\;\;$}~\color{black}), \added{the Yee scheme} with an all-fine mesh (\color{black} ~\st{$\;\;\;\;\;\;\;$}~\color{black}), subgridding (\color{legendGreen}~\st{$\;\;\;\;\;\;$}~\color{black}), \added{ADI-FDTD subgridding~\cite{YangADIFDTDsub2D}} (\color{legendLightgreen}~\st{$\;\;\;$}~\st{$\;\;\;$}~\color{black}), proposed method with only MOR~(\color{legendPurple}~\st{$\;\;\;$}~\st{$\;\;\;$}~\color{black}), proposed method with only CFL extension (\color{legendPink}~\st{$\;$}~\st{$\;$}~\st{$\;$}~\st{$\;$}~\color{black}), and proposed method with both MOR and CFL limit extension (\color{legendRed}~\st{$\;$}~\st{$\;\;\;$}~\st{$\;$}~\color{black}).}

	\label{fig:filter_results}
\end{figure}

\begin{table}[t]\centering 
	\centering
	\caption{Execution time for the waveguide structure in Sec.\ref{sec:ne_filter}.}
	\begin{tabular}{  l c c c r r }
		
		\hline 
		
		Method                                     &   CFL number & Runtime & Speed-up \\ \hline \hline
		\added{Yee scheme}, all-coarse grid                                  &  $0.99$  & 28.9 s   & - \\ \hline
		\added{Yee scheme}, all-fine grid  								 &  $0.99$ 	&  281.2 s & - \\ \hline 
		Subgridding                              &  $0.99$  & 222.0 s  & 1.26X  \\ \hline
	\added{ADI-FDTD Subgridding~\cite{YangADIFDTDsub2D}}				                  & $1.98$  & 262.2 s    & 1.07X\\ \hline
		Proposed (MOR	only)		             & $0.99$ 	&  245.4 s  & 1.15X\\ \hline
		Proposed (CFL extension	only)		             & $1.98$   & 484.3 s & 0.58X \\ \hline
		Proposed 			                 &    $1.98$  &  118.5 s & 2.37X\\ \hline

	\end{tabular}
	\label{Filter_table}
\end{table} 

\subsection{Reflection Test}
\label{sec:ne_rods4}

To further investigate the accuracy of the proposed scheme, we analyse the reflections from a scatterer placed in a waveguide. A 66~mm~$\times$~40~mm  waveguide terminated with 15~mm-thick perfectly matched layers on two sides is considered. The layout of the structure is shown  in Fig.~\ref{fig:rods4}. The scatterer is comprised of four small copper rods with 1~mm radius. A line source is used to excite the waveguide and a line probe near the source position is chosen. The structure is uniformly discretized with coarse cells with $\dx=\dy=$ 1~mm.   In the proposed method, the region where the four rods are located is refined with  $r = 6$. All methods were run with a time step equal to 0.99 times their CFL limit. \added{The proposed method and ADI-FDTD subgridding ~\cite{YangADIFDTDsub2D} were run with a time step three times larger than the CFL limit of the fine grid.} In the proposed method, MOR was used to reduce the number of variables in the fine region from 7,008 to 1,920. 

The reflected power measured with all methods is shown in the top panel of Fig.~\ref{fig:result_4rods}. A conventional FDTD run (\added{Yee scheme}) with a coarse grid everywhere overestimates reflections from the rods across the entire frequency range, confirming the need for a finer mesh around the scatterers. The proposed method and subgridding are instead in very good agreement with the reference simulation, which was performed with \added{the Yee scheme} and a fine mesh in the entire computational domain. In order to further assess the accuracy of the proposed method, we finally repeated the analysis of the waveguide without the four rods. The purpose of this test is to measure the level of reflections caused by the transition from the coarse grid to the reduced model of the fine grid. The reflections from the interface obtained with both subgridding and the proposed method are depicted in the bottom panel of Fig.~\ref{fig:result_4rods}. The figure shows that the proposed method can provide, for a given region, a reduced model with enhanced CFL limit that can be seamlessly embedded into an FDTD grid, without resulting in increased reflections or stability issues. 

\added{In terms of performance, the proposed method is 27.5 times} \added{faster than an all-fine FDTD simulation, 2.1 times faster than} \added{subgridding, and 3.8 times faster than ADI-FDTD subgridding, as shown in Table~\ref{table:result_rods4}. The proposed method achieves higher performance than ADI-FDTD subgridding because, like implicit methods, the proposed method is able to run above the CFL limit of the fine grid. However, its complexity is lower, since the fine grid is represented with a reduced-order model, and only explicit update operations are required at runtime.}

\begin{figure}[t]
	\centering
	\begin{tikzpicture}
	[label/.style={font=\fontsize{\sizefont}{\sizefont}\selectfont}]

\pgfmathsetmacro{\lenx}{5.28}
\pgfmathsetmacro{\leny}{3.2}
\pgfmathsetmacro{\PMLax}{1.2}
\pgfmathsetmacro{\PMLbx}{4.08}
\pgfmathsetmacro{\subgrax}{3.12}
\pgfmathsetmacro{\subgray}{1.28}
\pgfmathsetmacro{\subgrbx}{3.76}
\pgfmathsetmacro{\subgrby}{1.92}
\pgfmathsetmacro{\rodCoordAx}{3.28}
\pgfmathsetmacro{\rodCoordAy}{1.76}
\pgfmathsetmacro{\rodCoordBx}{3.28}
\pgfmathsetmacro{\rodCoordBy}{1.44}
\pgfmathsetmacro{\rodCoordCx}{3.6}
\pgfmathsetmacro{\rodCoordCy}{1.76}
\pgfmathsetmacro{\rodCoordDx}{3.6}
\pgfmathsetmacro{\rodCoordDy}{1.44}
\pgfmathsetmacro{\curx}{1.36}
\pgfmathsetmacro{\interfx}{1.52}
\pgfmathsetmacro{\rodR}{0.08}
\pgfmathtruncatemacro{\truelenxmm}{66}
\pgfmathtruncatemacro{\truelenymm}{40}
\pgfmathtruncatemacro{\dxCmm}{1}
\pgfmathtruncatemacro{\truedistInterfToSubgrmm}{20}
\pgfmathtruncatemacro{\truesubgrxmm}{8}
\pgfmathtruncatemacro{\truesubgrymm}{8}
\pgfmathtruncatemacro{\trueRodToRodNearestmm}{2}
\pgfmathtruncatemacro{\trueRodDiammm}{2}

	\pgfmathsetmacro{\sizefont} {9};
	
	\draw [fill = {rgb:black,25;green,50;blue,25}] (0, 0) rectangle (\lenx,\leny);
	\draw [fill = white] (\PMLax, 0) rectangle(\PMLbx, \leny);
	
	\draw[font=\fontsize{\sizefont}{\sizefont}\selectfont] (\PMLax/2,\leny/2) node {PML};
	\draw[font=\fontsize{\sizefont}{\sizefont}\selectfont] (\lenx - \PMLax/2,\leny/2) node {PML};
	
	\draw[dashed] (\subgrax, \subgray) rectangle (\subgrbx, \subgrby);
	\draw[label] (\subgrax/2 + \subgrbx/2, \subgrby) node [above] {\truesubgrxmm$\times$\truesubgrymm\ mm};
	
	\draw [fill = gray!50] (\rodCoordAx, \rodCoordAy) circle (\rodR);
	\draw [fill = gray!50] (\rodCoordBx, \rodCoordBy) circle (\rodR);
	\draw [fill = gray!50] (\rodCoordCx, \rodCoordCy) circle (\rodR);
	\draw [fill = gray!50] (\rodCoordDx, \rodCoordDy) circle (\rodR);
	
	\draw [thick, red] (\curx,0) -- (\curx, \leny);
	\draw [thick, gray!150] (\interfx,0) -- (\interfx, \leny);
	
	\draw [->,red,>=stealth, shorten >= 1pt] (\curx + 0.3, 0.85*\leny+0.2) -- (\curx, 0.85*\leny);
	\draw [label] (\curx + 0.3 -0.07, 0.85*\leny+0.2) node [right] {$J_y$ current};
	
	\draw [->,gray!150,>=stealth, shorten >= 1pt] (\interfx + 0.3, 0.73*\leny+0.2) -- (\interfx, 0.73*\leny);
	\draw [label] (\interfx + 0.3 -0.07, 0.73*\leny+0.2) node [right] {Probes};
	
	
	\draw [very thick] (0,0) -- (\lenx, 0);
	\draw [very thick] (0,\leny) -- (\lenx, \leny);
	\draw[label] (\lenx/2 - 0.6, \leny + 0.4) node {PEC};
	\draw[->,black,>=stealth, shorten >= 3pt, shorten <= 8pt] (\lenx/2 - 0.6, \leny + 0.4) -- (\lenx/2, \leny);
	
			\pgfmathsetmacro{\axisoffs}{\lenx/6}
			\draw[->,black,>=stealth] (\lenx + \axisoffs,\leny*0.7) -- (\lenx + \lenx/10+\axisoffs,\leny*0.7);
			\draw[->,black,>=stealth] (\lenx + \axisoffs,\leny*0.7) -- (\lenx + \axisoffs,\leny*0.7 + \lenx/10);
			\draw [fill=white] (\lenx + \axisoffs,\leny*0.7) circle (0.02*\lenx);
			\draw [fill=black] (\lenx + \axisoffs,\leny*0.7) circle (0.0075*\lenx);
			\draw[label] (\lenx + \axisoffs+\lenx/10,\leny*0.7) node [right] {x};
			\draw[label] (\lenx + \axisoffs,\leny*0.7 + \lenx/10) node [above] {y};
			\draw[label] (\lenx + \axisoffs,\leny*0.7) node [below left] {z};
	
	\pgfmathsetmacro{\d}{\leny/30};
	\draw[<->,black,>=stealth] (0,-\d) -- (\lenx, -\d);	
	\draw[label] (\lenx/2, -\d) node [below] {\truelenxmm\ mm};

	\draw[<->,black,>=stealth] (-\d,0) -- (-\d, \leny);	
	\draw[label] (-\d, \leny/2) node [rotate = 90, above] {\truelenymm\ mm};
	
	\draw[label] (\interfx,0) node [above right] {\begin{tabular}{l} $\dx$ = \dxCmm\ mm \\ $\dy$ = \dxCmm\ mm \end{tabular}};
	
	\draw[<->,black,>=stealth, shorten >= 0pt, shorten <= 0pt] (\interfx,\leny/2) -- (\subgrax, \leny/2);	
	\draw[label] (\interfx/2 + \subgrax/2, \leny/2) node [below] {\truedistInterfToSubgrmm\ mm};
	
	\pgfmathsetmacro{\offsx}{\lenx+\lenx/15};
	\pgfmathsetmacro{\offsy}{0};
	
	\pgfmathsetmacro{\subexpand}{2};
	\pgfmathsetmacro{\subgrlenx}{(\subgrbx - \subgrax)};
	\pgfmathsetmacro{\subgrleny}{(\subgrby - \subgray)};
	\pgfmathsetmacro{\subgrrodAx}{(\subexpand*\rodCoordAx - \subexpand*\subgrax + \offsx)};
	\pgfmathsetmacro{\subgrrodAy}{(\subexpand*\rodCoordAy - \subexpand*\subgray + \offsy)};
	\pgfmathsetmacro{\subgrrodBx}{(\subexpand*\rodCoordBx - \subexpand*\subgrax + \offsx)};
	\pgfmathsetmacro{\subgrrodBy}{(\subexpand*\rodCoordBy - \subexpand*\subgray + \offsy)};
	\pgfmathsetmacro{\subgrrodCx}{(\subexpand*\rodCoordCx - \subexpand*\subgrax + \offsx)};
	\pgfmathsetmacro{\subgrrodCy}{(\subexpand*\rodCoordCy - \subexpand*\subgray + \offsy)};
	\pgfmathsetmacro{\subgrrodDx}{(\subexpand*\rodCoordDx - \subexpand*\subgrax + \offsx)};
	\pgfmathsetmacro{\subgrrodDy}{(\subexpand*\rodCoordDy - \subexpand*\subgray + \offsy)};
	\draw[label] (\offsx + \subexpand*\subgrlenx/2, \offsy) node [below] {Scatterer};
	
\draw [fill = gray!50] (\subgrrodAx, \subgrrodAy) circle (\subexpand*\rodR);
\draw [fill = gray!50] (\subgrrodBx, \subgrrodBy) circle (\subexpand*\rodR);
\draw [fill = gray!50] (\subgrrodCx, \subgrrodCy) circle (\subexpand*\rodR);
\draw [fill = gray!50] (\subgrrodDx, \subgrrodDy) circle (\subexpand*\rodR);	
\draw [<->, >=stealth, shorten >= 2pt, shorten <= 2pt] (\subgrrodCx, \subgrrodCy-\rodR) -- (\subgrrodDx, \subgrrodDy+\rodR);
\draw [label] (\subgrrodCx, \subgrrodCy/2+\subgrrodDy/2) node [right] {\trueRodToRodNearestmm\ mm};
	\pgfmathsetmacro{\cos}{0.7071};
	\pgfmathsetmacro{\sin}{0.7071};
	 \draw [<->, >=stealth] (\subgrrodAx- \cos*\rodR*\subexpand, \subgrrodAy - \sin*\rodR*\subexpand) -- (\subgrrodAx+ \cos*\rodR*\subexpand, \subgrrodAy + \sin*\rodR*\subexpand);
	 
	 \draw [label] (\subgrrodAx+\cos*\rodR*\subexpand, \subgrrodAy + \sin*\rodR*\subexpand) node [above] {\trueRodDiammm\ mm}; 
	
	\end{tikzpicture}
	\caption{Layout of the four-rod reflection test in Sec.~\ref{sec:ne_rods4}. The dashed line shows the area modeled with a fine grid, and replaced by a reduced model in the proposed technique.}
	\label{fig:rods4}
\end{figure}
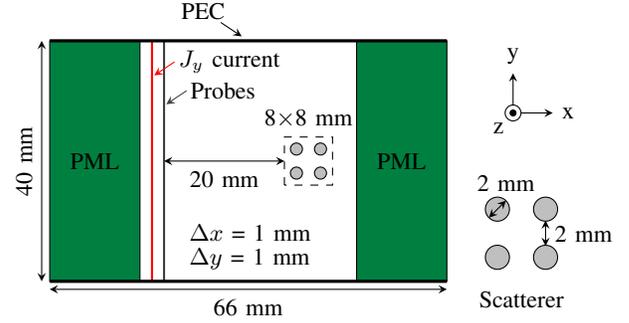

\begin{figure}[t]
	\centering

	\includegraphics[width=1\columnwidth]{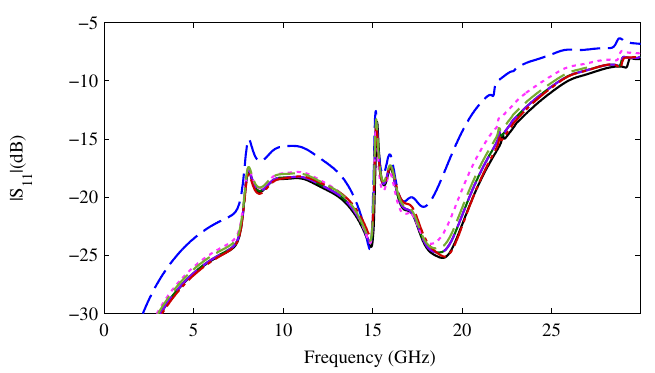}
	\includegraphics[width=1\columnwidth]{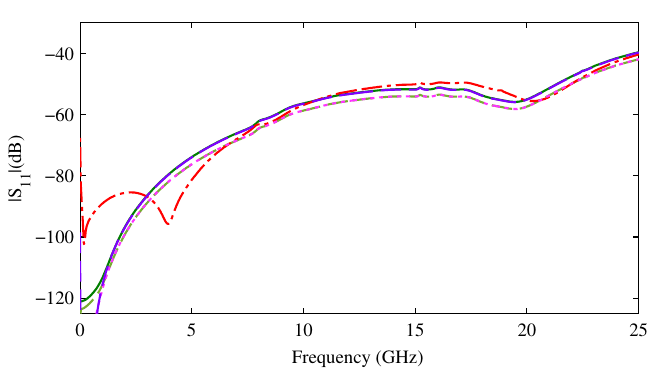}

	\caption{Reflected power with respect to the incident for the test in Sec.~\ref{sec:ne_rods4}. Top panel: reflections from the four-rod scatterer. Bottom panel: reflections with no rods inside. Results were computed with \added{the Yee scheme} with an all-coarse mesh~(\color{legendBlue} ~\st{$\;\;$}~\st{$\;\;\;$}~\st{$\;\;$}~\color{black}), an all-fine mesh (\color{black} ~\st{$\;\;\;\;\;\;\;$}~\color{black}), subgridding (\color{legendGreen}~\st{$\;\;\;\;\;\;$}~\color{black}), \added{ADI-FDTD subgridding~\cite{YangADIFDTDsub2D}} (\color{legendLightgreen}~\st{$\;\;\;$}~\st{$\;\;\;$}~\color{black}), the proposed method with only MOR (\color{legendPurple}~\st{$\;\;\;$}~\st{$\;\;\;$}~\color{black}), the proposed method with only CFL extension (\color{legendPink}~\st{$\;$}~\st{$\;$}~\st{$\;$}~\st{$\;$}~\color{black}), and the proposed method with both MOR and CFL limit extension (\color{legendRed}~\st{$\;$}~\st{$\;\;\;$}~\st{$\;$}~\color{black}).}

	\label{fig:result_4rods}
\end{figure}

\begin{table}[t]
	\centering
	\caption{Execution time for the four rods scattering test in Sec.~\ref{sec:ne_rods4}.}
	\begin{tabular}{ l l  l r r}
		\hline
		Method                              & CFL number& Runtime & Speed-up \\ \hline \hline
		\added{Yee scheme}, all-coarse grid                         & $0.99$  & 4.0 s     & -        \\ \hline 
		\added{Yee scheme}, all-fine grid   						  & $0.99$ & 654.2 s   & -        \\ \hline  

		Subgridding			                & $0.99$ 	&  50.5 s   & 12.9X   \\ \hline
	    \added{ADI-FDTD Subgridding~\cite{YangADIFDTDsub2D}}			                  & $2.97$  & 92.3 s &  7.1X\\ \hline
		Proposed (MOR only)		            & $0.99$    & 68.4 s    & 9.6X    \\ \hline 
		Proposed (CFL extension	only)		    & $2.97$	&  537.2 s   & 1.2X \\ \hline
		Proposed			            & $2.97$ 	&  23.8 s   & 27.5X   \\ \hline

	\end{tabular}
	\label{table:result_rods4}
\end{table}

\section{Conclusion}

In this paper, we proposed a stable FDTD scheme which supports the inclusion of reduced models to capture complex objects with fine features. Initially, complex objects are discretized with a refined mesh, in order to properly capture their geometry. Then, the FDTD equations for each refined region are compressed with model order reduction. The reduced models are finally coupled to the surrounding coarse mesh. The proposed method supports the embedding of multiple reduced models, losses and inhomogeneous material properties. The stability and CFL limit of the resulting scheme are rigorously proved, which is a novel result. The CFL limit can be also extended, with negligible accuracy loss, with a simple perturbation of the model coefficients. Numerical tests confirm the stability of the proposed method, and its potential to increase FDTD's efficiency for multiscale problems. Future works will extend the proposed ideas to the three-dimensional case.

\appendices
\section{\added{Proof of Theorem~\ref{thm:passivitycascade}}}
\label{sec:append}

\added{We prove Theorem~\ref{thm:passivitycascade} for the case of two FDTD-like subsystems being connected together using a fairly standard argument from the theory of dissipative systems~\cite{brogliato2006dissipative}. Since a system made by $N$ subsystems can be obtained by connecting one subsystem at a time, our proof demonstrates the Theorem for the general case.

We consider the configuration shown in Fig.~\ref{fig:sys1and2}. Two passive subsystems are connected to each other through some of their ports, while the other ports become the external ports of the resulting system, shown by the dashed box in Fig.~\ref{fig:sys1and2}. Since subsystems~1 and~2 are dissipative, they both individually satisfy the dissipation inequality~\cite{byrnes1994losslessness}
\begin{equation}
\E_1(\vect{x}_1^{n+1}) - \E_1(\vect{x}_1^{n}) \le s_1(\vect{y}_1^n, \vect{u}_1^{n+\half}) \,,
\label{eq:dissipineq1}
\end{equation}
\begin{equation}
\E_2(\vect{x}_2^{n+1}) - \E_2(\vect{x}_2^{n}) \le s_2(\vect{y}_2^n, \vect{u}_2^{n+\half})\,,
\label{eq:dissipineq2}
\end{equation}
where:
\begin{itemize}
	\item $\E_1(\vect{x}_1^{n})$ and $\E_2(\vect{x}_2^{n})$ are the storage functions of subsystem~1 and~2, respectively, which gives the energy stored inside each subsystem at time $n$;
	\item  $\vect{u}_1^{n+\half}$  and $\vect{y}_1^{n}$ are, respectively, the input and output of subsystem~1, and can be partitioned as
	\begin{equation}
	\vect{u}_1^{n+\half} = \begin{bmatrix}
	\vect{u}_{1,A}^{n+\half}\\
	\vect{u}_{1,B}^{n+\half}\\
	\end{bmatrix}\,, \qquad
	\vect{y}_1^{n} = \begin{bmatrix}
	\vect{y}_{1,A}^{n}\\
	\vect{y}_{1,B}^{n}\\
	\end{bmatrix}\\ \,,
	\end{equation}
	where vectors $\vect{u}_{1,A}^{n+\half}$ and $\vect{y}_{1,A}^{n}$ correspond to unconnected ports, while $\vect{u}_{1,B}^{n+\half}$ and $\vect{y}_{1,B}^{n}$ correspond to the ports connected to subsystem~2;
	\item $\vect{u}_2^{n+\half}$ and $\vect{y}_2^{n}$ are similarly defined for subsystem~2, and can be partitioned in the same way
	\begin{equation}
	\vect{u}_2^{n+\half} = \begin{bmatrix}
	\vect{u}_{2,B}^{n+\half}\\
	\vect{u}_{2,C}^{n+\half}\\
	\end{bmatrix}\,,\qquad
	\vect{y}_2^{n} = \begin{bmatrix}
	\vect{y}_{2,B}^{n}\\
	\vect{y}_{2,C}^{n}\\
	\end{bmatrix}\,,
	\end{equation}
  where $\vect{u}_{2,B}^{n+\half}$ and $\vect{y}_{2,B}^{n}$ correspond to connected ports, while $\vect{u}_{2,C}^{n+\half}$ and $\vect{y}_{2,C}^{n}$ correspond to unconnected ports;
	\item $s_1(.,.)$ and $s_2(.,.)$ are the supply rates of subsystem~1 and subsystem~2 respectively, which give the energy entering each subsystem between time $n$ and $n+1$, as discussed in~\cite{bekmambetova2016dissipative}. Taking into account the partitioning of the inputs and outputs of the two subsystems, the two supply rates can be written as
	\begin{equation}
	s_1(\vect{y}_1^n, \vect{u}_1^{n+\half}) =  s_{1,A}(\vect{y}_{1,A}^n, \vect{u}_{1,A}^{n+\half})+ s_{1,B}(\vect{y}_{1,B}^n, \vect{u}_{1,B}^{n+\half})\,,
	\end{equation}
	\begin{equation}
	s_2(\vect{y}_2^n, \vect{u}_2^{n+\half}) =  s_{2,B}(\vect{y}_{2,B}^n, \vect{u}_{2,B}^{n+\half})+ s_{2,C}(\vect{y}_{2,C}^n, \vect{u}_{2,C}^{n+\half})\,.
	\end{equation}
\end{itemize}
\begin{figure}[t]
	\centering
	\begin{tikzpicture}
	[scale = 1]

	\pgfmathsetmacro{\sizefont} {9}
	
	\pgfmathsetmacro{\W}{1.95}
	\pgfmathsetmacro{\H}{2}
	\pgfmathsetmacro{\Rbig}{0.08}
	
	\draw (1.5*\W,0) -- (1.5*\W,1.2*\H); 
	\draw (1.5*\W,0) -- (2.5*\W,0*\H);
	\draw (1.5*\W,1.2*\H) -- (2.5*\W,1.2*\H);
	\draw (2.5*\W,0) -- (2.5*\W,1.2*\H);

	\draw  (-1*\W,0.6*\H) circle (\Rbig);
	\draw (-1*\W,0.6*\H) -- (-0.5*\W,0.6*\H);
	\draw  (-1*\W,0.4*\H) circle (\Rbig);
	\draw (-1*\W,0.4*\H) -- (-0.5*\W,0.4*\H);
	
	\draw (0.5*\W,0.3*\H) -- (0.8*\W,0.3*\H);
	\draw (0.5*\W,0.9*\H) -- (0.8*\W,0.9*\H);
	\draw (-1*\W,0.6*\H) -- (-1*\W,0.6*\H);
	
	\draw (-0.5*\W,0) -- (0.5*\W,0);
	\draw (-0.5*\W,1.2*\H) -- (0.5*\W,1.2*\H);
	\draw (-0.5*\W,0) -- (-0.5*\W,1.2*\H);
	\draw (0.5*\W,0*\H) -- (0.5*\W,1.2*\H);
	
	\draw (0.5*\W,0.3*\H) -- (1.5*\W,0.3*\H);
	\draw (0.5*\W,0.9*\H) -- (1.5*\W,0.9*\H);
	\draw (0.5*\W,0.6*\H) -- (1.5*\W,0.6*\H);

	\draw  (0.8*\W, 0.3*\H) circle (\Rbig);
	
	\draw  (0.8*\W, 0.9*\H) circle (\Rbig);
	
	\draw (1.2*\W,0.3*\H) circle (\Rbig);
	\draw (1.2*\W,0.9*\H) circle (\Rbig);

	\draw (1.2*\W,0.6*\H) circle (\Rbig);
	\draw (0.8*\W,0.6*\H) circle (\Rbig);

	\draw (2.5*\W,0.8*\H) -- (3.2*\W,0.8*\H);
	\draw (3.2*\W,0.8*\H) circle (\Rbig);		
	\draw (2.5*\W,0.6*\H) -- (3.2*\W,0.6*\H);
	\draw (3.2*\W,0.6*\H) circle (\Rbig);
	\draw (2.5*\W,0.4*\H) -- (3.2*\W,0.4*\H);
	\draw (3.2*\W,0.4*\H) circle (\Rbig);

	\draw[label] (-0.15*\W,0.58*\H) node [above] {~~~~Subsystem 1 };
	\draw[label] (2.0*\W,0.58*\H) node [above] {~Subsystem 2 };
	\draw[label] (-1*\W,0.6*\H) node [above]{$\vect{u}_{1,A}^{n+\half}$};
	\draw[label] (0.8*\W,0.9*\H) node [above]{$\vect{u}_{1,B}^{n+\half}$};
	\draw[label] (1.2*\W,0.9*\H) node [above]{$\vect{u}_{2,B}^{n+\half}$};
	\draw[label] (2.8*\W,0.8*\H) node [above]{$\vect{u}_{2,C}^{n+\half}$};
	\draw[dashed] (-0.8*\W,-0.4*\H) rectangle (3*\W,1.5*\H);
	
	\end{tikzpicture}
	
	\caption{\added{Configuration considered in Appendix~\ref{sec:append} to prove Theorem~\ref{thm:passivitycascade}.}}
		\label{fig:sys1and2}

\end{figure}

Once the two subsystems are connected together, the cascaded system will have as inputs and outputs
	\begin{equation}
	\vect{u}^{n+\half} = \begin{bmatrix}
	\vect{u}_{1,A}^{n+\half}\\
	\vect{u}_{2,C}^{n+\half} \\
	\end{bmatrix} \,,
	\vect{y}^{n} = \begin{bmatrix}
	\vect{y}_{1,A}^{n}\\
	\vect{y}_{2,C}^{n} \\
	\end{bmatrix} \,,
	\end{equation}
	as also visible in Fig.~\ref{fig:sys1and2}.
	In order to prove that the cascaded system is passive, we have to define its own storage function $\E(\vect{x}^{n})$, supply rate $s(\vect{y}^n, \vect{u}^{n+\half})$, and prove that they satisfy the dissipation inequality~\cite{byrnes1994losslessness}
	\begin{equation}
		\E(\vect{x}^{n+1}) - \E(\vect{x}^{n}) \le s(\vect{y}^n, \vect{u}^{n+\half})\,.
		\label{eq:dissipineq}
	\end{equation}
	As storage function, we take 
	\begin{equation}
	\E(\vect{x}^{n}) = \E_1(\vect{x}_1^{n})+\E_2(\vect{x}_2^{n}) \,,
	\end{equation}
	since obviously the energy stored in the cascaded system is the sum of the energy stored in each subsystem. The supply rate is defined as 
	\begin{multline}
	s(\vect{y}^n, \vect{u}^{n+\half}) =  
	s_{1,A}(\vect{y}_{1,A}^n, \vect{u}_{1,A}^{n+\half})+s_{2,C}(\vect{y}_{2,C}^n, \vect{u}_{2,C}^{n+\half})\,,
	\label{eq:supplySys}
    \end{multline}
    because the energy entering into the cascaded system is the sum of the energy entering from the leftmost ports and the rightmost ports in Fig.~\ref{fig:sys1and2}. 
    
    We now prove~\eqref{eq:dissipineq}. Adding~\eqref{eq:dissipineq1} and~\eqref{eq:dissipineq2} side by side, we obtain
      \begin{multline}
    \E(\vect{x}^{n+1}) - \E(\vect{x}^{n})  \le s_1(\vect{y}_1^n, \vect{u}_1^{n+\half}) + s_2(\vect{y}_2^n, \vect{u}_2^{n+\half}) \\
    = s_{1,A}(\vect{y}_{1,A}^n, \vect{u}_{1,A}^{n+\half})+s_{1,B}(\vect{y}_{1,B}^n, \vect{u}_{1,B}^{n+\half}) \\
    +s_{2,B}(\vect{y}_{2,B}^n, \vect{u}_{2,B}^{n+\half})+s_{2,C}(\vect{y}_{2,C}^n, \vect{u}_{2,C}^{n+\half})\,.
    \label{eq:chain}
   \end{multline}
	Using the expression for the supply rate of an FDTD-like system given in~\cite{bekmambetova2016dissipative}, one can easily show that
	\begin{equation}
	s_{1,B}(\vect{y}_{1,B}^n, \vect{u}_{1,B}^{n+\half}) 
	+s_{2,B}(\vect{y}_{2,B}^n, \vect{u}_{2,B}^{n+\half}) = 0 \,,
	\label{eq:cancelSys}
	\end{equation}
	which physically means that all the energy leaving subsystem~1 from the connected ports will enter subsystem~2. Substituting~\eqref{eq:cancelSys} into~\eqref{eq:chain}, we obtain~\eqref{eq:dissipineq}. The cascaded system is then passive and, consequently, does not have any unstable poles outside the unit circle in the complex plane~\cite{hitz1969discrete}. }

\bibliographystyle{IEEEtran}
\bibliography{IEEEabrv,bibliography}

\begin{IEEEbiography}[{\includegraphics[width=1in,height=1.25in,clip,keepaspectratio]{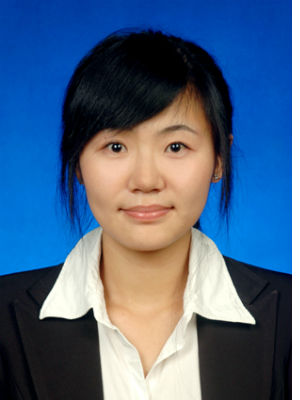}}]{Xinyue Zhang} (S'16) received her B.Sc. and M.S. degree in electrical engineering from Harbin Institute of Technology, Harbin, China, in 2012 and 2014, respectively. Currently, she is working towards a Ph.D. degree in electrical engineering from the University of Toronto, Toronto, ON, Canada. Her research interests include computational electromagnetics, model order reduction, and wireless communication. She received the Best Paper Award of 6th International ICST Conference on Communications and Networking, 2011.
\end{IEEEbiography}

\begin{IEEEbiography}[{\includegraphics[width=1in,height=1.25in,clip,keepaspectratio]{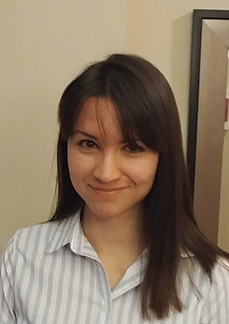}}]{Fadime Bekmambetova} (S'16)
	completed the B.A.Sc degree in Engineering Science at the University of Toronto in 2016 and is currently pursuing the M.A.Sc degree at the same university. In 2016, she received the Best Student Paper Award at the 25th IEEE Conference on Electronic Packages and Systems. 
\end{IEEEbiography}

\begin{IEEEbiography}[{\includegraphics[width=1in,height=1.25in,clip,keepaspectratio]{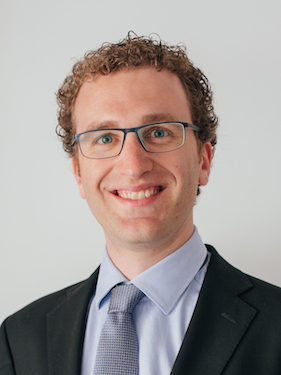}}]{Piero Triverio} (S'06 -- M'09 -- SM'16)
	received the M.Sc. and Ph.D. degrees in Electronic Engineering from Politecnico di Torino, Italy in 2005 and 2009, respectively. He is an Assistant Professor with the Department of Electrical and Computer Engineering at the University of Toronto, where he holds the Canada Research Chair in Modeling of Electrical Interconnects.
	His research interests include signal integrity, electromagnetic compatibility, and model order reduction. He received several international awards, including the 2007 Best Paper Award of the IEEE Transactions on Advanced Packaging, the EuMIC Young Engineer Prize at the 13th European Microwave Week, and the Best Paper Award at the IEEE 17th Topical Meeting on Electrical Performance of Electronic Packaging.
\end{IEEEbiography}

\end{document}